\documentclass[10pt,sigconf,letterpaper,nonacm]{acmart}

\usepackage[english]{babel}
\usepackage{blindtext}
\usepackage{float}
\usepackage{enumitem}
\usepackage{fancyhdr}
\usepackage{algorithm}
\usepackage{float}
\usepackage{amsmath, amsthm}
\usepackage{graphicx}
\usepackage{subcaption}
\usepackage{booktabs}
\usepackage{microtype}
\usepackage{algpseudocode}
\usepackage{amsmath}
\usepackage{microtype}
\newtheorem{proposition}{Proposition}
\usepackage{array}
\usepackage{booktabs}
\renewcommand\footnotetextcopyrightpermission[1]{} 
\setcopyright{none}

\acmDOI{}

\acmISBN{}

\acmPrice{}

\begin{document}
\emergencystretch=2em
\title{Rethinking Traffic Matrix Completion: Estimate the Process, Not the Entries}


\author{Xiyuan Liu}
\authornote{Xiyuan Liu and Zihao Wang contributed equally to this work.}
\affiliation{%
  \institution{Xidian University}
  \city{Xi'an}
  \country{China}}

\author{Zihao Wang}
\authornotemark[1]
\email{zihaowang23@stu.xidian.edu.cn}
\affiliation{%
  \institution{Xidian University}
  \city{Xi'an}
  \country{China}}

\author{Guanzuo Liu}
\email{guanzuoliu@stu.xidian.edu.cn}
\affiliation{%
  \institution{Xidian University}
  \city{Xi'an}
  \country{China}}

\author{Xiucheng Tian}
\email{xiuchengtian@stu.xidian.edu.cn}
\affiliation{%
  \institution{Xidian University}
  \city{Xi'an}
  \country{China}}

\author{Wenting Wei}
\authornote{Corresponding author.}
\email{wtwei@xidian.edu.cn}
\affiliation{%
  \institution{Xidian University}
  \city{Xi'an}
  \country{China}}

\renewcommand{\shortauthors}{Liu, Wang, et al.}
\begin{abstract}
Traffic matrix measurement is fundamental for datacenter operations, but obtaining complete traffic matrices at scale remains challenging due to the prohibitive cost of global fine-grained measurement and partial observations resulting from network faults. Although existing matrix completion methods (reduce cost) achieve satisfactory performance in specific scenarios, their reliance on  restrictive assumptions or  black-box mappings  results in a lack of interpretability and an inability to characterize uncertainty.
In this paper, we propose Utimac, an uncertainty-aware traffic matrix completion for data center networks. Our analysis shows that, within a locally stationary window, log-domain traffic can be decomposed into a principal statistical component and a sparse deviation component. Based on this insight, we formulate traffic matrix completion as a parameter inference problem: multiple partially observed frames within a window are used to infer shared parameters and recover missing entries.  To avoid the intractability and boundary degeneracy of the original integral-form marginal likelihood, we construct a regularized surrogate objective and solve the resulting joint optimization problem with block coordinate descent.
Utimac consistently outperforms all baselines on data center networks datasets in both overall and burst scenarios, with its advantage becoming more pronounced as observations grow sparser.
All code is publicly available in an anonymous repository: https://anonymous.4open.science/r/Utimac-0551/
\end{abstract}

\keywords{Traffic Matrix Completion, Uncertainty-aware}
\maketitle

\section{Introduction}
In recent years, datacenter networks have been supporting increasingly large-scale workloads \cite{xia2016survey}, and applications such as distributed training, LLM online inference, and large-scale data processing have further intensified network communication \cite{gangidi2024rdma,qian2024alibaba}. As these workloads continue to expand, accurate measurement of  network traffic becomes important \cite{canel2024understanding}, because it provides critical data support for downstream tasks such as traffic engineering, fault localization, and anomaly detection \cite{benson2010network}.

Obtaining a complete traffic matrix in a large-scale datacenter is challenging. As the network grows, exhaustively measuring all source-destination pairs requires devices to maintain large flow tables, which continuously consume the limited memory and processing resources of switches \cite{li2016flowradar}. Meanwhile, the available observations are often incomplete \cite{gursun2012traffic}. Fine-grained monitoring is typically enabled only at selected locations and during selected periods, and the collected statistics can further become incomplete due to reporting delays, record loss, or device failures \cite{zhang2017high}.

Existing datacenter traffic measurement methods can be divided into direct measurement and indirect measurement. Direct measurement obtains traffic statistics by collecting packets or flows from network devices or monitoring points; representative examples include flow-recording approaches such as NetFlow/IPFIX and systems like OpenTM \cite{claise2004cisco,claise2013specification,tootoonchian2010opentm}, which directly measure traffic matrices using counters in switch flow tables. In addition, sketch-based designs such as UnivMon and ElasticSketch maintain compact summaries to approximately record traffic features with lower memory overhead \cite{liu2016one,yang2018elastic}. 
Indirect measurement infers the global traffic matrix from partial observations. Some of the works rely on explicit priors such as low-rankness or sparsity, as illustrated by Xie et al.’s studies on variable-rate measurements \cite{xie2018accurate}, block matrix completion \cite{xie2021low}, and deep adversarial tensor completion \cite{xie2023deep}. While some adopt data-driven modeling and learn the mapping from partial observations to the full traffic matrix, with representative examples including AutoTomo and Satformer \cite{qiao2024autotomo,qin2024satformer}.

Unfortunately, existing methods still have several gaps. For direct measurement, the main limitation is that the measurement overhead grows with network scale and measurement granularity\cite{ghabashneh2022microscopic}. For example, flow-record and per-flow counting methods require devices to maintain a large amount of flow-level tables \cite{yu2013software}, while sampling and path-counting methods introduce extra overhead for path maintenance or statistics aggregation \cite{sekar2008csamp}. As a result, direct measurement often leads to high resource and management costs in large-scale data center networks.
In contrast, indirect measurement reduces the cost of explicit monitoring, but its performance often depends on strong assumptions or the training data. Methods based on priors such as low rankness rely on strong assumptions \cite{roughan2011spatio}, and their performance degrades when real traffic deviates from these assumptions. Deep learning based methods can learn complex mappings, but they usually lack interpretability and are sensitive to distribution shifts, which limits their generalization. Most existing indirect measurement methods focus on point estimation and cannot provide confidence ranges for the traffic matrix. However, traffic matrix imputation is inherently an inference problem under incomplete observations, and its outputs therefore contain uncertainty. Without explicitly characterizing this uncertainty, downstream systems cannot identify high-risk predictions and may be misled in subsequent decision-making \cite{mei2023uncertainty}.

We propose Utimac, an traffic matrix completion method that mind the gaps of existing methods: the high cost of direct measurement, the limited interpretability of black-box inference, and the lack of uncertainty characterization in point estimation. Our data analysis shows that, within a locally stationary time window, log-domain traffic can be described by the superposition of a joint Gaussian principal component and a Laplacian deviation component. Based on this statistical structure, we formulate traffic completion as a parameter inference problem driven by partial observations: we infer the shared mean, shared covariance, and sparsity parameter from multiple partially observed frames within a window and then recover the unobserved traffic entries accordingly. 

Our contributions are summarized as follows:
\begin{itemize}
    \item We reformulate traffic matrix completion as a statistical inference problem. Within each locally stationary window, we represent log-domain traffic by a joint Gaussian principal component and a Laplacian deviation component, and tie multiple partially observed frames to a shared set of parameters. Completion is then carried out by inferring the shared mean, covariance, and sparsity parameter from partial observations and recovering the missing entries.
    
    \item We derive a computable and well-posed inference objective from the integral-form marginal likelihood. To avoid high-dimensional integration and eliminate the degeneracy caused by sparsity growth and covariance collapse, we introduce a surrogate objective together with stabilizing regularizers, which yields a regularized nested optimization model.
    
    \item We develop a structured solver for the regularized problem. We rewrite it as a joint optimization over deviation variables, mean, sparsity, and covariance, and solve it by block coordinate descent. This procedure alternates between frame-wise deviation updates, mean estimation, sparsity update, and covariance optimization, and finally recovers the missing traffic entries together with their uncertainty.
\end{itemize}

Utimac consistently outperforms all baselines on data center networks datasets in both overall and burst scenarios, with its advantage becoming more pronounced as observations grow sparser.

\section{Problem Statement}

This section establishes the foundation of our formulation. We first introduce the system model and assumptions, and then formalize the traffic matrix completion problem under partial observations.

\subsection{System Model and Assumptions}
\label{subsec:system_model}

We model data center traffic in discrete time. At the selected network layer, let
$\{u_1,u_2,\dots,u_M\}$ denote the source-side network units and
$\{v_1,v_2,\dots,v_N\}$ denote the destination-side network units.
For each source-destination pair $(u_i,v_j)$, let $X_t(i,j)\in\mathbb{R}_{+}$ be the traffic random variable at time $t$, and let $x_t(i,j)$ be one realization of $X_t(i,j)$.

By vectorizing all source-destination flows at time $t$ in a fixed order, we obtain the traffic random vector:
\begin{equation}
\begin{aligned}
X_t
&=
\begin{bmatrix}
X_t(s_1) & X_t(s_2) & \cdots & X_t(s_d)
\end{bmatrix}^{\top}
\in\mathbb{R}_{+}^{d}, d=MN,
\end{aligned}
\label{eq:Xt_def}
\end{equation}
where $S=\{s_1,s_2,\dots,s_d\}$ is the spatial index set after vectorization, and each index $s_k$ corresponds to one unique source-destination pair. The corresponding realization is denoted by $x_t\in\mathbb{R}_{+}^{d}$.

We adopt a locally stationary model. Let the time axis be partitioned into local windows. $n$-th window is denoted by:
\begin{equation}
W_n=\{t_{n,1},t_{n,2},\dots,t_{n,L_n}\},
\label{eq:window_def}
\end{equation}
where $L_n=|W_n|$. We assume that the distribution parameters of the traffic process remain approximately constant within the same window $W_n$, and may change across different windows.

To account for non-negativity and heavy-tailed behavior, we work in the log domain. Let $\varepsilon\in\mathbb{R}_{++}^{d}$ be a strictly positive offset vector. We define
\begin{equation}
Z_t=\log(X_t+\varepsilon),
\label{eq:Zt_def}
\end{equation}
where the logarithm is applied elementwise.

We decompose the log-domain traffic into a structured component and a deviation component:
\begin{equation}
Z_t=U_t+O_t,\qquad t\in W_n.
\label{eq:decomp}
\end{equation}

Here, $U_t\in\mathbb{R}^{d}$ represents the dominant statistical structure, and $O_t\in\mathbb{R}^{d}$ represents local deviations.

Later empirical analysis shows that, within a local stationary window, the dominant part of the log-domain traffic exhibits an approximately joint Gaussian pattern, while the deviation part is sparse. Based on that empirical finding, we use the following parametric model in window $W_n$:
\begin{equation}
U_t\sim\mathcal{N}(\mu_n,\Sigma_n),
\qquad t\in W_n,
\label{eq:Ut_gaussian}
\end{equation}
and
\begin{equation}
p_O(o;\lambda_n)
=
\left(\frac{\lambda_n}{2}\right)^d
\exp\!\left(-\lambda_n\|o\|_1\right),
\label{eq:Ot_laplace}
\end{equation}
where $\mu_n\in\mathbb{R}^{d}$ is the shared mean vector, $\Sigma_n\in\mathbb{R}^{d\times d}$ is the shared covariance matrix with $\Sigma_n\succ 0$, and $\lambda_n>0$ is the shared sparsity parameter. We use
\begin{equation}
\theta_n=(\mu_n,\Sigma_n,\lambda_n)
\label{eq:theta_def}
\end{equation}
to denote the shared statistical parameters in window $W_n$.

We further assume conditional independence within each window. Given $\theta_n$, the pairs $\{(U_t,O_t):t\in W_n\}$ are independent across time. For each $t\in W_n$, $U_t$ and $O_t$ are independent. The joint density of the latent variables at time $t$ is therefore
\begin{equation}
\begin{aligned}
p_{U,O}(u_t,o_t;\theta_n)
&=
p_U(u_t;\mu_n,\Sigma_n)\,
p_O(o_t;\lambda_n) \\
&\propto
\exp\!\left(
-\frac{1}{2}(u_t-\mu_n)^{\top}\Sigma_n^{-1}(u_t-\mu_n)
-\lambda_n\|o_t\|_1
\right).
\end{aligned}
\label{eq:joint_latent_density}
\end{equation}

For each time $t$, let $b_t\in\{0,1\}^{d}$ be the binary observation mask. The observed index set is
\begin{equation}
\Omega_t=\{s_k\in S:\,b_t(s_k)=1\}.
\label{eq:Omega_t_def}
\end{equation}

The actual observation at time $t$ is the subvector of $x_t$ on $\Omega_t$, denoted by $x_t^{(\Omega_t)}\in\mathbb{R}_{+}^{|\Omega_t|}$.

We treat $\Omega_t$ as exogenous and known during parameter estimation. All subsequent inference is conditioned on $\Omega_t$. Under this setting, the main task in each local window is to estimate the shared parameter $\theta_n$ from partial observations and then characterize the conditional distribution of the complete traffic vector.

\subsection{Problem Formulation}
\label{subsec:problem_formulation}

The main goal is to estimate the shared parameter $\theta_n$ in each local stationary window. Once $\theta_n$ is estimated, the distribution of the complete traffic vector in that window is specified. Given current observations, the unobserved entries are characterized by their conditional posterior distribution.

For any time $t\in W_n$, define the observed log-domain subvector
\begin{equation}
z_t^{(\Omega_t)}
=
\log\!\left(x_t^{(\Omega_t)}+\varepsilon_{\Omega_t}\right)
\in\mathbb{R}^{|\Omega_t|},
\label{eq:zt_obs_def}
\end{equation}
where $\varepsilon_{\Omega_t}$ is the subvector of $\varepsilon$ indexed by $\Omega_t$. By \eqref{eq:decomp}, we have
\begin{equation}
z_t^{(\Omega_t)}
=
u_t^{(\Omega_t)}+o_t^{(\Omega_t)}.
\label{eq:obs_decomp}
\end{equation}

Since $U_t\sim\mathcal{N}(\mu_n,\Sigma_n)$, its marginal distribution on any observed index set remains Gaussian:
\begin{equation}
U_t^{(\Omega_t)}
\sim
\mathcal{N}\!\left(
\mu_n^{(\Omega_t)},
\Sigma_n^{(\Omega_t,\Omega_t)}
\right),
\label{eq:Ut_obs_gaussian}
\end{equation}
where $\mu_n^{(\Omega_t)}$ is the subvector of $\mu_n$ on $\Omega_t$, and $\Sigma_n^{(\Omega_t,\Omega_t)}$ is the principal submatrix of $\Sigma_n$ indexed by $\Omega_t$. The deviation term on the same subspace has density
\begin{equation}
p_O\!\left(o_t^{(\Omega_t)};\lambda_n\right)
=
\left(\frac{\lambda_n}{2}\right)^{|\Omega_t|}
\exp\!\left(-\lambda_n\|o_t^{(\Omega_t)}\|_1\right).
\label{eq:Ot_obs_density}
\end{equation}

Given $\theta_n$, the exact conditional likelihood of the observed log-domain vector is obtained by marginalizing out the latent deviation variable:
\begin{equation}
\begin{aligned}
p\!\left(z_t^{(\Omega_t)}\mid \Omega_t;\theta_n\right)
=
\int_{\mathbb{R}^{|\Omega_t|}}
&p_U\!\left(
z_t^{(\Omega_t)}-o;
\mu_n^{(\Omega_t)},
\Sigma_n^{(\Omega_t,\Omega_t)}
\right) \\
&\cdot p_O(o;\lambda_n)\,do .
\end{aligned}
\label{eq:single_frame_marginal}
\end{equation}

By conditional independence within $W_n$, the principled observed-data maximum likelihood estimator is
\begin{equation}
\begin{aligned}
\hat{\theta}_n^{\mathrm{ML}}
=
\arg\max_{\mu_n,\;\Sigma_n\succ 0,\;\lambda_n>0}
\sum_{t\in W_n}
\log p\!\left(
z_t^{(\Omega_t)}\mid
\Omega_t;
\mu_n,\Sigma_n,\lambda_n
\right).
\end{aligned}
\label{eq:window_exact_ml}
\end{equation}

The objective in \eqref{eq:window_exact_ml} is statistically well defined. Its direct evaluation is expensive because each frame requires the computation of a high-dimensional convolution integral. For a general non-diagonal covariance matrix $\Sigma_n^{(\Omega_t,\Omega_t)}$, the integral in \eqref{eq:single_frame_marginal} does not admit a simple closed form.

To obtain a computable objective, we introduce a profiled approximation. Using the identity
\begin{equation}
u_t^{(\Omega_t)}=z_t^{(\Omega_t)}-o_t^{(\Omega_t)},
\label{eq:u_eliminate}
\end{equation}
the reduced joint log-density for a feasible decomposition can be written as
\begin{equation}
\begin{aligned}
\tilde{\mathcal{L}}_t\!\left(\theta_n,o_t^{(\Omega_t)}\right)
{}&=
-\frac{1}{2}
\left(z_t^{(\Omega_t)}-o_t^{(\Omega_t)}-\mu_n^{(\Omega_t)}\right)^{\top}
\left(\Sigma_n^{(\Omega_t,\Omega_t)}\right)^{-1} \\
&\cdot
\left(z_t^{(\Omega_t)}-o_t^{(\Omega_t)}-\mu_n^{(\Omega_t)}\right)
-\frac{1}{2}\log\det\Sigma_n^{(\Omega_t,\Omega_t)} \\
&+|\Omega_t|\log\lambda_n
-\lambda_n\|o_t^{(\Omega_t)}\|_1
+C_t,
\end{aligned}
\label{eq:reduced_log_joint}
\end{equation}
where $C_t$ is constant optimization variables.

We approximate the frame-level marginal log-likelihood by the maximum value of the reduced joint log-density. The resulting estimator is
\begin{equation}
\begin{aligned}
\hat{\theta}_n
=
\arg\max_{\mu_n,\;\Sigma_n\succ 0,\;\lambda_n>0}
\sum_{t\in W_n}
\max_{o_t^{(\Omega_t)}\in\mathbb{R}^{|\Omega_t|}}
\tilde{\mathcal{L}}_t\!\left(\theta_n,o_t^{(\Omega_t)}\right).
\end{aligned}
\label{eq:window_profiled_max}
\end{equation}

Dropping constants and changing the sign gives the equivalent minimization form:
\begin{equation}
\begin{aligned}
\hat{\theta}_n
=
\arg\min_{\mu_n,\;\Sigma_n\succ 0,\;\lambda_n>0}
\sum_{t\in W_n}
\min_{o_t^{(\Omega_t)}\in\mathbb{R}^{|\Omega_t|}}
\Bigg\{-
\tilde{\mathcal{L}}_t\!\left(\theta_n,o_t^{(\Omega_t)}\right)
\Bigg\}.
\end{aligned}
\label{eq:window_profiled_min}
\end{equation}

The objective in \eqref{eq:window_profiled_min} is a bilevel optimization problem. The outer variables are the shared parameters $(\mu_n,\Sigma_n,\lambda_n)$. The inner variables are the frame-specific deviation vectors $\{o_t^{(\Omega_t)}\}_{t\in W_n}$. For fixed outer parameters, each inner problem is convex because it contains a positive-definite quadratic term and an $\ell_1$ term. The full problem is generally nonconvex in $(\mu_n,\Sigma_n,\lambda_n)$ due to the inverse covariance term and the $\log\det$ term.

After obtaining $\hat{\theta}_n=(\hat{\mu}_n,\hat{\Sigma}_n,\hat{\lambda}_n)$, the conditional distribution of the unobserved entries is determined. Let
\begin{equation}
\Omega_t^{\mathrm{mis}}=S\setminus\Omega_t
\label{eq:Omega_mis_def}
\end{equation}
be the unobserved index set at time $t$. The posterior distribution in the log domain is
\begin{equation}
\begin{aligned}
p\!\left(
z_t^{(\Omega_t^{\mathrm{mis}})}
\mid
z_t^{(\Omega_t)},\Omega_t;\hat{\theta}_n
\right).
\end{aligned}
\label{eq:posterior_z}
\end{equation}

By the inverse transform
\begin{equation}
X_t=\exp(Z_t)-\varepsilon,
\label{eq:inverse_transform}
\end{equation}
we obtain the corresponding posterior distribution in the original traffic domain:
\begin{equation}
\begin{aligned}
p\!\left(
x_t^{(\Omega_t^{\mathrm{mis}})}
\mid
x_t^{(\Omega_t)},\Omega_t;\hat{\theta}_n
\right).
\end{aligned}
\label{eq:posterior_x}
\end{equation}
\section{Method}
\label{sec:method}

This section presents the log-domain Gaussian-Laplacian decomposition model, its regularized reformulation, a block coordinate descent solver, and uncertainty quantification.
\subsection{Statistical Structure Discovery}
\label{subsec:statistical_pattern}
The first observation is local stationarity over short intervals. Over such intervals, active jobs or resource contention change little, so consecutive frames are driven by similar communication patterns. As a result, the empirical mean, covariance, and overall distribution shape remain stable within a proper window, which supports using a local window for parameter sharing. 

The second observation concerns the dominant structure in the log domain. Raw traffic spans a wide dynamic range because flows differ substantially in message size. The log transform compresses this variation and regularizes the dominant structure. After transformation, most samples cluster around a stable center within each window, source-destination pairs remain statistically dependent due to common jobs or similar communication patterns. A detailed empirical analysis demonstrating that the dominant component of log-domain traffic approximately follows a joint Gaussian distribution is detailed in Appendix~\ref{appendix:gaussian}.

The third observation is the presence of sparse deviations around the dominant structure. Short-lived events such as incast bursts or transient congestion can create large deviations on a small number of entries. These events do not alter the dominant structure, but they produce localized perturbations that cannot be absorbed by the principal component. This motivates a sparse additive deviation term. Empirical evidence supporting the presence of these sparse deviations is presented in Appendix~\ref{app:sparse_deviation_evidence}.

Based on these observations, we model log-domain traffic in each local window as the sum of a dominant structured component and a sparse deviation component. We use a joint Gaussian model for the dominant component to capture the shared center and spatial dependence of regular traffic, and a Laplacian prior for the deviation component to capture sparse local perturbations. Accordingly, for each time $t \in W_n$, we write
\begin{equation}
Z_t = \log(X_t + \varepsilon) = U_t + O_t,
\label{eq:pattern_decomposition}
\end{equation}
where $U_t$ denotes the dominant structured component and $O_t$ denotes the sparse deviation component. We then parameterize them as
\begin{equation}
U_t \sim \mathcal{N}(\mu_n, \Sigma_n), \qquad t \in W_n,
\label{eq:pattern_gaussian}
\end{equation}
and
\begin{equation}
p_O(o;\lambda_n)
=
\left(\frac{\lambda_n}{2}\right)^d
\exp\!\left(-\lambda_n \|o\|_1\right),
\label{eq:pattern_laplacian}
\end{equation}
where $\mu_n$ and $\Sigma_n$ are the shared mean and covariance in window $W_n$, and $\lambda_n$ controls the sparsity of local deviations.

\subsection{Analysis and Reformulation}
\label{subsec:regularized_reformulation}

The profiled objective in \eqref{eq:window_profiled_min} is computable, but it is not directly suitable for optimization. To analyze its behavior, define
\begin{equation}
\Psi(\mu_n,\Sigma_n,\lambda_n)
:=
\sum_{\tau\in W_n}
\min_{o_\tau^{(\Omega_\tau)}}
\left\{
-\tilde{\mathcal{L}}_\tau\!\left(\theta_n,o_\tau^{(\Omega_\tau)}\right)
\right\}.
\label{eq:profile_obj_sec42}
\end{equation}

For fixed $(\mu_n,\Sigma_n,\lambda_n)$, each inner problem is a convex optimization over the deviation variable. The main difficulty lies in the outer objective, which is generally nonconvex in $(\mu_n,\Sigma_n,\lambda_n)$.

More importantly, \eqref{eq:profile_obj_sec42} is not lower bounded. First, since $o_\tau^{(\Omega_\tau)}=0$ is always feasible, we have
\begin{equation}
\Psi(\mu_n,\Sigma_n,\lambda_n)
\le
C(\mu_n,\Sigma_n)
-
\Bigg(\sum_{\tau\in W_n}|\Omega_\tau|\Bigg)\log\lambda_n,
\label{eq:lambda_escape_sec42}
\end{equation}
where $C(\mu_n,\Sigma_n)$ is independent of $\lambda_n$. Hence the objective decreases without bound as $\lambda_n\to\infty$.

Second, for any fixed $\mu_n$ and $\lambda_n>0$, one can choose
\begin{equation}
o_\tau^{(\Omega_\tau)}
=
z_\tau^{(\Omega_\tau)}-\mu_n^{(\Omega_\tau)},
\label{eq:sigma_o_sec42}
\end{equation}
which makes the quadratic residual term vanish. If we further set $\Sigma_n=\epsilon I$ with $\epsilon\downarrow 0$, then
\begin{equation}
\log\det \Sigma_n^{(\Omega_\tau,\Omega_\tau)}
=
|\Omega_\tau|\log\epsilon
\to -\infty.
\label{eq:sigma_collapse_sec42}
\end{equation}

This shows that the covariance matrix can continue to shrink and further reduce the objective through the log-determinant term.

To remove these two escape directions, we introduce the regularizers
\begin{equation}
\rho\,\mathrm{tr}(\Sigma_n^{-1}), \qquad \rho>0,
\label{eq:cov_reg_sec42}
\end{equation}
and
\begin{equation}
\eta\lambda_n, \qquad \eta>0.
\label{eq:lambda_reg_sec42}
\end{equation}

The resulting regularized problem is
\begin{equation}
(\hat{\mu}_n,\hat{\Sigma}_n,\hat{\lambda}_n)
=
\arg\min_{\mu_n,\Sigma_n\succ 0,\lambda_n>0}
\Psi(\mu_n,\Sigma_n,\lambda_n)
+\rho\,\mathrm{tr}(\Sigma_n^{-1})
+\eta\lambda_n .
\label{eq:regularized_obj_sec42}
\end{equation}

Equation \eqref{eq:regularized_obj_sec42} preserves the original nested inference structure while eliminating boundary degeneracy. It remains nonconvex and nonsmooth, but it is now better posed and numerically stable.

\subsection{Block Coordinate Solution}
\label{subsec:bcd_solver}

We solve \eqref{eq:regularized_obj_sec42} using block coordinate descent. By explicitly reintroducing the deviation variables, we rewrite the problem as
\begin{equation}
\min_{\mu_n,\Sigma_n\succ 0,\lambda_n>0,\{o_\tau^{(\Omega_\tau)}\}_{\tau\in W_n}}
F_n\!\left(\mu_n,\Sigma_n,\lambda_n,\{o_\tau^{(\Omega_\tau)}\}_{\tau\in W_n}\right),
\label{eq:joint_obj_sec43}
\end{equation}
where
\begin{equation}
F_n
=
\sum_{\tau\in W_n}
\phi_\tau(\mu_n,\Sigma_n,\lambda_n; o_\tau^{(\Omega_\tau)})
+\rho\,\mathrm{tr}(\Sigma_n^{-1})
+\eta\lambda_n .
\label{eq:joint_obj_expanded_sec43}
\end{equation}

This form exposes four variable blocks: deviation variables, mean, sparsity parameter, and covariance.

We first update the deviation variables. For fixed $\mu_n$, $\Sigma_n$, and $\lambda_n$, define
\begin{equation}
r_\tau=z_\tau^{(\Omega_\tau)}-\mu_n^{(\Omega_\tau)},\qquad
Q_\tau=\left(\Sigma_n^{(\Omega_\tau,\Omega_\tau)}\right)^{-1},
\label{eq:ot_aux_sec43}
\end{equation}
and solve, for each $\tau\in W_n$,
\begin{equation}
\begin{aligned}
\hat{o}_\tau^{(\Omega_\tau)}
=
\arg\min_{o_\tau^{(\Omega_\tau)}}
\bigg \{
\frac{1}{2}
\left(o_\tau^{(\Omega_\tau)}-r_\tau\right)^{\top}
Q_\tau
\left(o_\tau^{(\Omega_\tau)}-r_\tau\right)\\
+\lambda_n\|o_\tau^{(\Omega_\tau)}\|_1
\bigg \}.
\label{eq:ot_update_sec43}
\end{aligned}
\end{equation}

After that, we update the shared mean. Let $P_\tau$ be the selection matrix for $\Omega_\tau$, and define
\begin{equation}
a_\tau=z_\tau^{(\Omega_\tau)}-\hat{o}_\tau^{(\Omega_\tau)}.
\label{eq:mu_aux_sec43}
\end{equation}

Then $\mu_n$ is obtained from
\begin{equation}
\left(
\sum_{\tau\in W_n}
P_\tau^{\top}
\left(\Sigma_n^{(\Omega_\tau,\Omega_\tau)}\right)^{-1}
P_\tau
\right)\mu_n
=
\sum_{\tau\in W_n}
P_\tau^{\top}
\left(\Sigma_n^{(\Omega_\tau,\Omega_\tau)}\right)^{-1}
a_\tau .
\label{eq:mu_normal_sec43}
\end{equation}

The sparsity parameter is then updated by
\begin{equation}
\hat{\lambda}_n
=
\frac{\sum_{\tau\in W_n}|\Omega_\tau|}
{\sum_{\tau\in W_n}\|\hat{o}_\tau^{(\Omega_\tau)}\|_1+\eta}.
\label{eq:lambda_update_closed_sec43}
\end{equation}

Finally, define the residual
\begin{equation}
e_\tau^{(\Omega_\tau)}
=
z_\tau^{(\Omega_\tau)}
-\hat{o}_\tau^{(\Omega_\tau)}
-\mu_n^{(\Omega_\tau)},
\label{eq:sigma_resid_sec43}
\end{equation}
and update the covariance by

\begin{equation}
\begin{split}
\hat{\Sigma}_n = &\arg\min_{\Sigma_n\succ 0} \bigg\{ \sum_{\tau\in W_n} \frac{1}{2} [(e_\tau^{(\Omega_\tau)})^{\top} (\Sigma_n^{(\Omega_\tau,\Omega_\tau)})^{-1} e_\tau^{(\Omega_\tau)} \\+& \frac{1}{2}\log\det \Sigma_n^{(\Omega_\tau,\Omega_\tau)} ] 
+ \rho\,\mathrm{tr}(\Sigma_n^{-1}) \bigg\}.
\end{split}
\end{equation}

The complete iteration is therefore
\begin{equation}
\{o_\tau^{(\Omega_\tau)}\}_{\tau\in W_n}
\;\longrightarrow\;
\mu_n
\;\longrightarrow\;
\lambda_n
\;\longrightarrow\;
\Sigma_n .
\label{eq:bcd_order_sec43}
\end{equation}

This block structure makes the regularized nested inference problem tractable in practice.

\subsection{Uncertainty Quantification}
Once the shared parameter estimate
$\hat{\theta}_n = (\hat{\mu}_n, \hat{\Sigma}_n, \hat{\lambda}_n)$
and the frame-level optimal deviations
$\{\hat{o}_t^{(\Omega_t)}\}_{t \in W_n}$
have been obtained from the block coordinate descent procedure
in~(43), the conditional distribution of each unobserved log-domain
entry $Z_t(j)$, $j \in \Omega_t^{\mathrm{mis}}$, follows a
one-dimensional Normal--Laplace distribution arising from the
Gaussian--Laplace convolution posterior, from which a
closed-form marginal CDF enables efficient construction of
plug-in $95\%$ credible intervals in the original traffic domain;
the complete derivation is provided in Appendix~\ref{appendix:uncertainty}.

 
\section{Experiments}
\label{sec:experiments}
 
\subsection{Experimental Setup}
\label{subsec:setup}
 
We evaluate on real-world traffic datasets.
Facebook-Pod-B and Facebook-ToR-A are DCN datasets
from Facebook's production infrastructure:
Pod-B captures pod-level aggregated traffic ($N{=}8$, $d{=}64$),
while ToR-A captures rack-level traffic ($N{=}155$, $d{=}24{,}025$).
DCN traffic exhibits strong burstiness and heavy-tailed
distributions~\cite{benson2010network}, making these the primary
evaluation datasets; additional WAN results on GÉANT~\cite{uhlig2006providing}
($N{=}22$, $d{=}484$) are reported in Appendix~\ref{app:extended}.
We compare against three baselines covering complementary paradigms.
PSW-I~\cite{wang2025optimal} minimizes an optimal-transport
discrepancy between time-series patches without parametric training.
ImputeFormer~\cite{nie2024imputeformer} is a low-rank-induced
Transformer combining matrix-completion priors with projected attention
and Fourier sparsity regularization.
Diffusion-TM~\cite{yuan2025diffusion} is a DDPM-based
generative model using a routing-free TMC inference branch.
Utimac runs on an Intel Core i9-13980HX CPU; all neural baselines
run in PyTorch on an NVIDIA GeForce RTX~5090 GPU.
Entries are masked under
$p_{\mathrm{obs}} \in \{0.3, 0.4, \ldots, 0.9\}$;
the same fixed-seed mask is shared by all methods.
Burst flows are identified with dominance multiplier $\alpha{=}2.0$
and threshold $\beta{=}0.8$.
 
\subsection{Evaluation Metrics}
\label{subsec:metrics}
 
All metrics are computed exclusively on missing positions
$\Omega^{\mathrm{miss}}_t$ (see Appendix~\ref{app:metrics} for
formal definitions).
Overall imputation accuracy is measured by MAE, RMSE, and
wMAPE across all missing entries.
Burst flow detection uses Precision, Recall, and F1 to
assess identification of anomalously large flows within each time
slot; full per-dataset detection tables are provided in
Appendix~\ref{app:extended}.
Burst flow imputation accuracy uses Burst-MAE, Burst-RMSE,
and Burst-wMAPE to evaluate numerical recovery on true missing burst
flows, with Burst-wMAPE as the headline indicator.
Metric for accuracy of the estimation intervals is PICP$_{97.5}$, the fraction of true missing values falling within
the closed-form 97.5\% estimation interval derived.
 
\subsection{Results}
\label{subsec:results}
 
\begin{figure}[t]
  \centering
  \includegraphics[width=0.9\columnwidth]{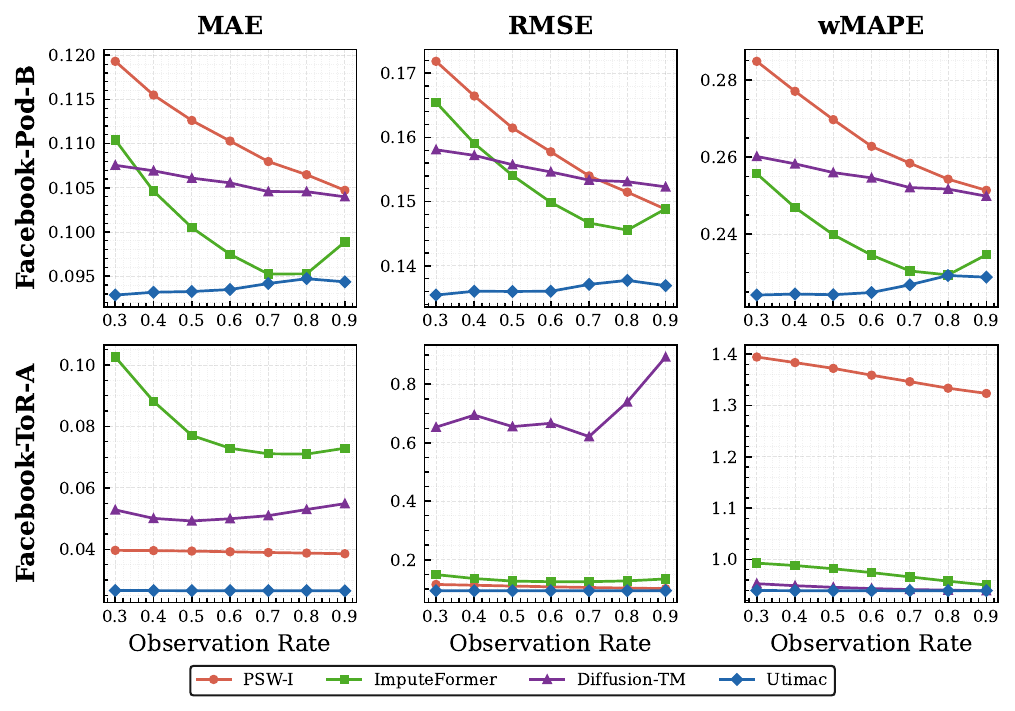}
  \caption{\textnormal{Overall MAE, RMSE, and wMAPE vs.\
    $p_{\mathrm{obs}}$ on Facebook-Pod-B (top) and
    Facebook-ToR-A (bottom).}}
  \label{fig:overall_metrics}
\end{figure}
 
\begin{figure}[t]
  \centering
  \includegraphics[width=0.9\columnwidth]{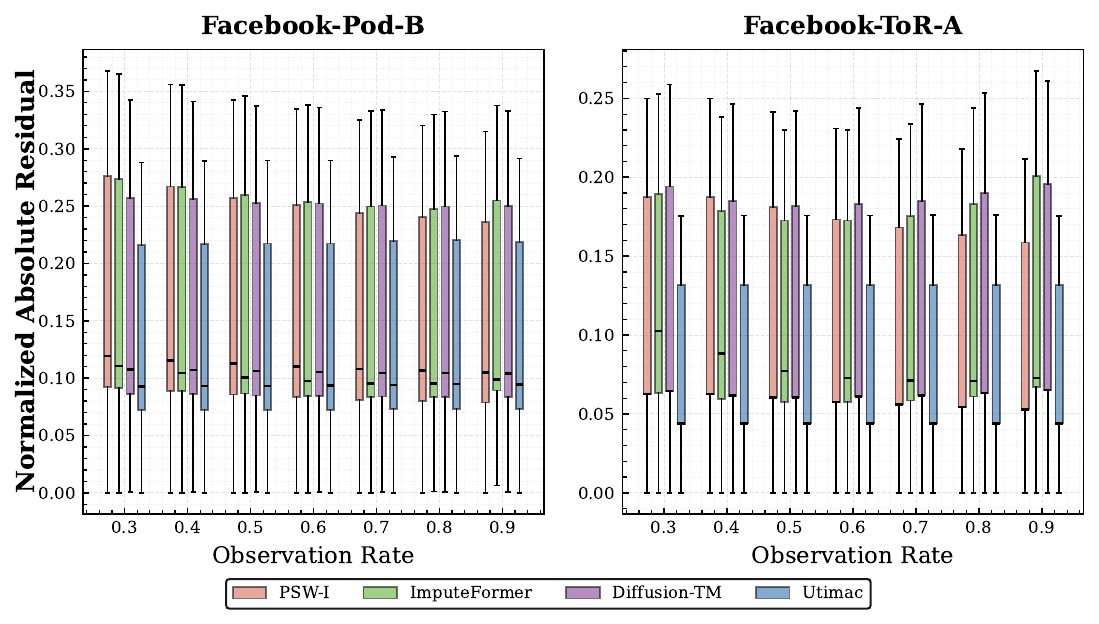}
  \caption{\textnormal{Normalized absolute residual distributions
    (5th--95th percentile) vs.\ $p_{\mathrm{obs}}$ on
    Facebook-Pod-B (left) and Facebook-ToR-A (right).}}
  \label{fig:residual_boxplot}
\end{figure}
 
\begin{figure}[t]
  \centering
  \includegraphics[width=0.9\columnwidth]{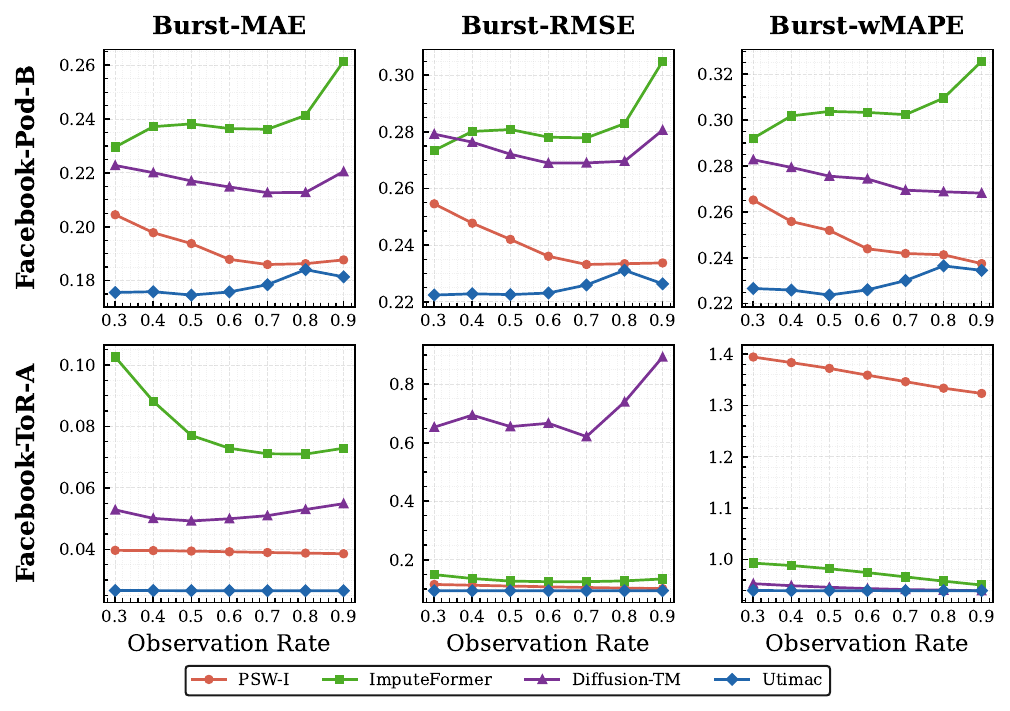}
  \caption{\textnormal{Burst-MAE, Burst-RMSE, and Burst-wMAPE
    vs.\ $p_{\mathrm{obs}}$ on Facebook-Pod-B and
    Facebook-ToR-A.}}
    \vspace{-0.5cm}
  \label{fig:burst_metrics}
\end{figure}
 
Figure~\ref{fig:overall_metrics} reports overall MAE, RMSE, and wMAPE
on DCN datasets.
On Facebook-Pod-B, Utimac leads all baselines across the full
observation range: its wMAPE advantage over the best competing method
(ImputeFormer) is $\mathbf{12.3\%}$ at $p_{\mathrm{obs}}{=}0.3$ and
narrows to $2.5\%$ at $p_{\mathrm{obs}}{=}0.9$.
This trend directly reflects the theoretical guarantee established in
Section~\ref{sec:method}: by analytically characterizing the
statistical structure of DCN traffic, Utimac
constructs a statistically principled completion that is most
informative precisely when observations are scarce, while
learning-based methods require more supervision to close the gap.
On Facebook-ToR-A, the lead is decisive and stable:
Utimac's MAE is $\mathbf{32.9\%}$ below PSW-I at
$p_{\mathrm{obs}}{=}0.3$ and remains $31.3\%$ lower at
$p_{\mathrm{obs}}{=}0.9$.
Diffusion-TM exhibits severe RMSE instability on this dataset,
reaching $0.89$ at $p_{\mathrm{obs}}{=}0.9$ ($9.5{\times}$ Utimac).
 
Figure~\ref{fig:residual_boxplot} confirms that Utimac's lower mean
error reflects uniformly better per-entry recovery: its 95th-percentile
residual on Facebook-Pod-B is $0.288$ at $p_{\mathrm{obs}}{=}0.3$,
versus $0.368$ for PSW-I and $0.343$ for Diffusion-TM.
Figure~\ref{fig:burst_metrics} evaluates recovery on true missing burst
flows ($\alpha{=}2.0$, $\beta{=}0.8$).
On Facebook-Pod-B, Utimac's Burst-wMAPE advantage over PSW-I
is $\mathbf{14.5\%}$ at $p_{\mathrm{obs}}{=}0.3$, narrowing to
$1.2\%$ at $p_{\mathrm{obs}}{=}0.9$, confirming that the inference
advantage is most valuable when observations are scarce.
On Facebook-ToR-A, Utimac maintains the lowest Burst-wMAPE
throughout; ImputeFormer shows the weakest burst recovery at
medium-to-high $p_{\mathrm{obs}}$ on Pod-B. Accuracy of the estimation intervals is further validated via PICP$_{97.5}$ in Appendix~\ref{app:picp}.

\section{Conclusion}
This paper presents Utimac, an uncertainty-aware traffic matrix completion method that models log-domain traffic as a joint Gaussian principal component plus a Laplacian sparse deviation, casting completion as parameter inference solved by block coordinate descent with closed-form predictive intervals. Experiments on real-world DCN and WAN datasets show that Utimac consistently outperforms representative baselines, with its advantage most pronounced under sparse observations.

\bibliographystyle{ACM-Reference-Format}
\bibliography{reference}

\appendix

\section{Empirical Validation of Joint Gaussianity for Log-Domain Traffic Vectors}
\label{appendix:gaussian}

This appendix validates the model assumption that the principal component approximation of the decomposition of the logarithmic domain flow vector follows a Gaussian joint distribution
$\mathcal{N}(\boldsymbol{\mu}_n,\boldsymbol{\Sigma}_n)$,
using a locally stationary window of $L = 200$ frames and
$d = 56$ OD pairs from the Facebook pod-b dataset.

\subsection{Validation Criterion and Strategy}
\label{appendix:gaussian:criterion}

The necessary and sufficient condition for joint Gaussianity is given
by the Cram\'{e}r--Wold theorem.

\begin{proposition}[Cram\'{e}r--Wold: Necessary and Sufficient Condition]
\label{prop:cramer_wold}
  $\boldsymbol{Z} \in \mathbb{R}^d$ follows
  $\mathcal{N}(\boldsymbol{\mu}, \boldsymbol{\Sigma})$
  if and only if for every nonzero $\boldsymbol{a} \in \mathbb{R}^d$,
  \begin{equation}
    \boldsymbol{a}^\top \boldsymbol{Z}
    \;\sim\;
    \mathcal{N}\!\bigl(\boldsymbol{a}^\top\boldsymbol{\mu},\;
    \boldsymbol{a}^\top\boldsymbol{\Sigma}\boldsymbol{a}\bigr).
    \label{eq:cramer_wold}
  \end{equation}
\end{proposition}

Directly verifying Proposition~\ref{prop:cramer_wold} requires applying
a one-dimensional normality test to every direction in $\mathbb{R}^d$,
which is statistically infeasible.
We therefore adopt a two-stage progressive strategy:
first, a global screening is performed using a necessary condition
for joint Gaussianity;
then, Proposition~\ref{prop:cramer_wold} is approximately verified on a
finite, representative set of directions.

A necessary condition for joint Gaussianity follows from the theoretical
distribution of the Mahalanobis distance:
if $\boldsymbol{Z} \sim \mathcal{N}(\boldsymbol{\mu}, \boldsymbol{\Sigma})$, then
\begin{equation}
  D_M^2(\boldsymbol{z})
  = (\boldsymbol{z} - \hat{\boldsymbol{\mu}})^\top
    \hat{\boldsymbol{\Sigma}}^{-1}
    (\boldsymbol{z} - \hat{\boldsymbol{\mu}})
  \;\sim\; \chi^2(d).
  \label{eq:mahal}
\end{equation}

The ordered squared Mahalanobis distances
\begin{equation}
  D_{M,(1)}^2 \leq \cdots \leq D_{M,(n)}^2
  \label{eq:ordered_mahal}
\end{equation}
are plotted against the theoretical quantiles
\begin{equation}
  q_k = F_{\chi^2(d)}^{-1}\!\left(\frac{k-0.5}{n}\right),
  \quad k = 1, \ldots, n,
  \label{eq:chi2_quantile}
\end{equation}
to form the QQ plot.
If the point cloud systematically deviates from the reference line
$y = x$, the joint Gaussian hypothesis is directly rejected;
if the data points closely track the reference line, the necessary
condition in~\eqref{eq:mahal} is supported, though this result does
not imply sufficiency.
Section~\ref{appendix:gaussian:proj} then approximately verifies
the one-dimensional Gaussianity required by~\eqref{eq:cramer_wold}
along 181 representative directions.

\subsection{Mahalanobis Distance QQ Plot}
\label{appendix:gaussian:mahal}

Figure~\ref{fig:mahal_qq} compares the Mahalanobis distance QQ plots
for the raw domain and the log domain.
In the raw domain, the point cloud exhibits a pronounced S-shaped
deviation: the lower tail falls below the reference line $y = x$,
and a small number of points in the upper tail lie significantly above
the reference line (observed values near 92 at the theoretical quantile
around 80), reflecting the heavy-tailed, non-Gaussian nature of raw traffic.
In the log domain, the point cloud closely tracks the reference line
across the entire value range, with only minor deviations in the
extreme lower tail, indicating that the log transformation effectively
corrects the distributional shape to be compatible with joint Gaussianity.

\begin{figure}[htbp]
  \centering
  \begin{subfigure}[b]{0.47\columnwidth}
    \centering
    \includegraphics[width=\textwidth]{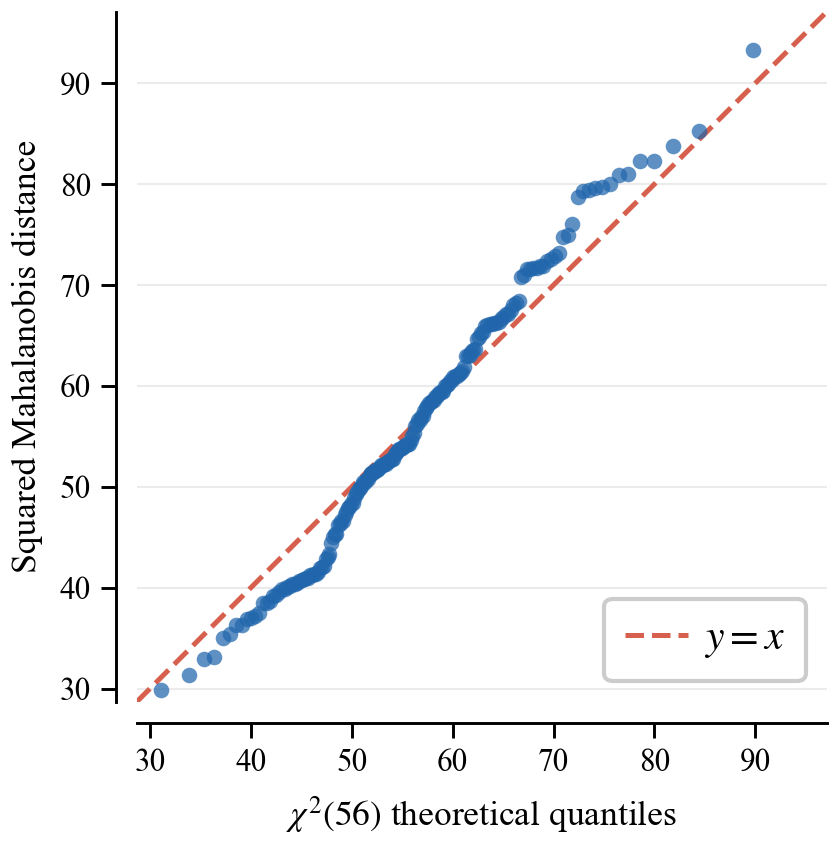}
    \caption{Raw domain}
    \label{fig:mahal_raw}
  \end{subfigure}
  \hfill
  \begin{subfigure}[b]{0.47\columnwidth}
    \centering
    \includegraphics[width=\textwidth]{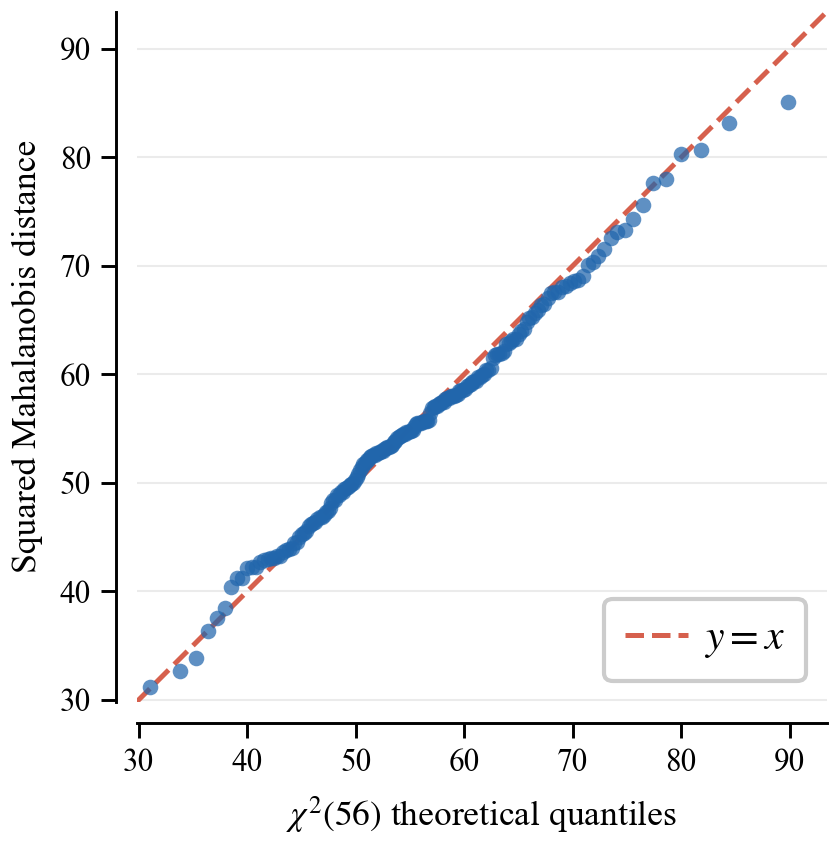}
    \caption{Log domain}
    \label{fig:mahal_log}
  \end{subfigure}
  \caption{\textnormal{%
    Mahalanobis distance QQ plots for the raw and log domains.
  }}
  \label{fig:mahal_qq}
\end{figure}

\subsection{Multi-Direction Projection Tests}
\label{appendix:gaussian:proj}

To approximately verify the necessary and sufficient condition in
Proposition~\ref{prop:cramer_wold}, projection sequences
$s_t = \boldsymbol{a}^\top \boldsymbol{z}_t$
are computed on standardized log-domain data along 181 representative
directions and subjected to the
Shapiro--Wilk test~\cite{shapiro1965analysis},
the D'Agostino $K^2$ test~\cite{d1973tests},
and the Anderson--Darling test~\cite{anderson1954test}
($\alpha = 0.05$).
The tested directions cover four categories:
coordinate directions $\boldsymbol{e}_k$ (56~directions),
the top 5 eigenvectors of the sample covariance matrix
(PCA, 5~directions),
pairwise combination directions
$(\boldsymbol{e}_i \pm \boldsymbol{e}_j)/\sqrt{2}$ (20~directions),
and random directions sampled uniformly from the unit sphere
(100~directions).

As the most representative subset for
Proposition~\ref{prop:cramer_wold},
the PCA directions achieve a pass rate of 4/5 (80\%),
the random directions achieve 84\% (84/100),
and the two categories combined yield 88/105 = 83.8\%.
The coordinate directions have a lower pass rate (17/56, 30.4\%),
consistent with individual OD pairs exhibiting local heavy-tail
perturbations caused by short-term traffic bursts;
such perturbations are averaged out in PCA and randomly mixed
directions, allowing the joint Gaussian structure to emerge.

Figure~\ref{fig:proj_qq} presents QQ plots for six representative
passing directions, covering the PCA, coordinate, and random
categories.
The point clouds closely track the OLS best-fit line in each panel,
consistent with the one-dimensional Gaussianity required
by~\eqref{eq:cramer_wold}.

\begin{figure}[t]
  \centering
  \includegraphics[width=\columnwidth]{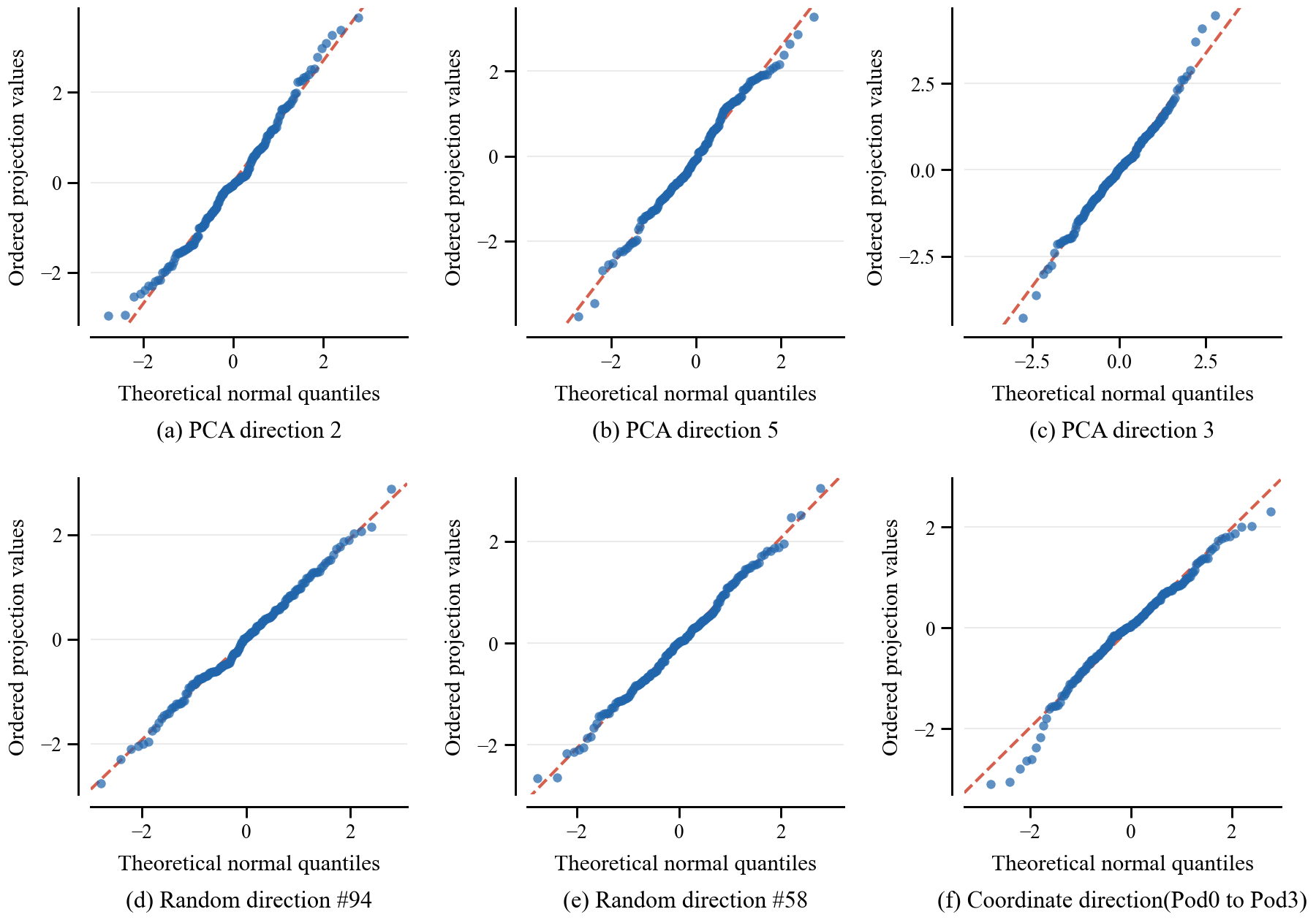}
  \caption{\textnormal{%
    Normal QQ plots for six representative projection directions in the
    log domain.
  }}
  \label{fig:proj_qq}
\end{figure}

\begin{figure}[htbp]
  \centering
  \includegraphics[width=\columnwidth]{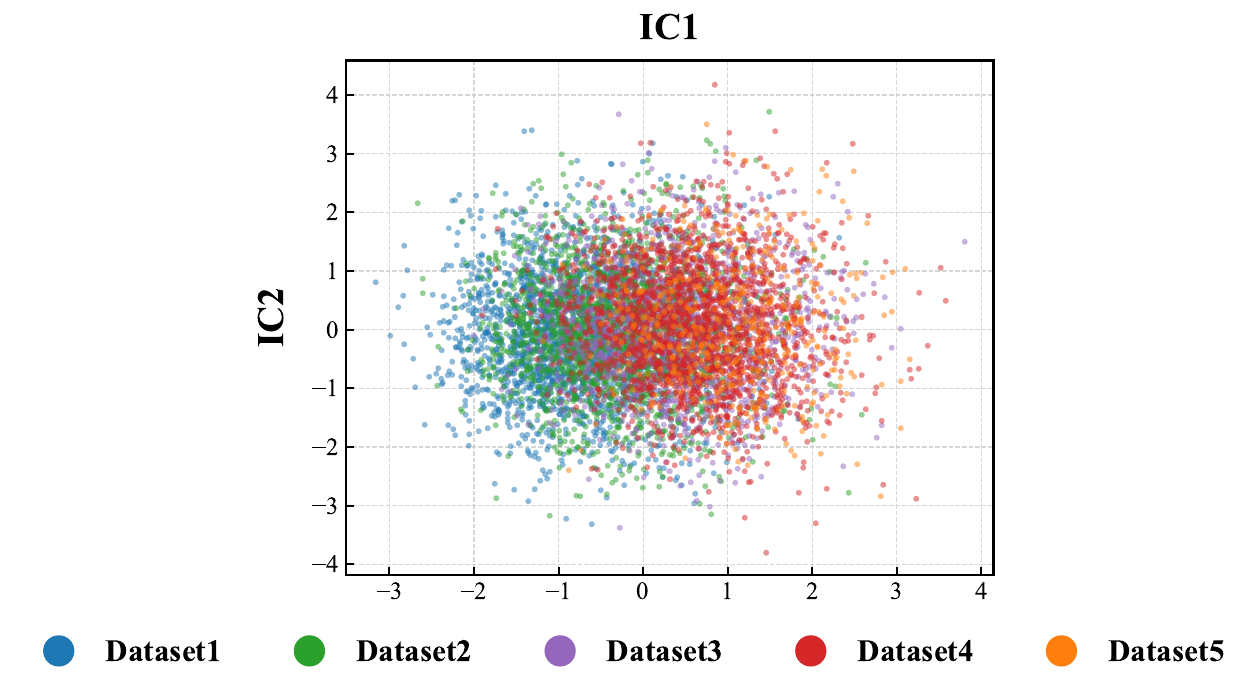}
  \caption{\textnormal{%
    FastICA projection of Facebook-Pod-B traffic onto IC1 and IC2,
    colour-coded by dataset split (three training sets and two
    validation sets).
    The majority of samples form a dense cluster near the origin,
    corresponding to the dominant low-magnitude traffic structure;
    a small number of scattered outliers represent sparse
    but off-center flows.}}
  \label{fig:fastica_clusters}
\end{figure}

\section{Empirical Evidence for Sparse Deviations around the Dominant Structure}
\label{app:sparse_deviation_evidence}

This appendix provides empirical evidence for the observation in the
main text that sparse deviations exist around a dominant traffic
structure. We examine the statistical behavior of the raw traffic data
from two complementary perspectives: low-dimensional projection and
distributional fitting. The results show that datacenter traffic is
dominated by a large number of small flows and a small number of
sparse but off-center flows.

We conduct a FastICA experiment on Facebook-Pod-B, projecting
traffic samples from all five dataset splits (three training and
two validation) jointly onto the first two independent components
(IC1 and IC2), colour-coded by split identity.

As shown in Figure~\ref{fig:fastica_clusters}, samples from all
five splits overlap in a compact elliptical cluster centred at the
origin, confirming that the dominant traffic structure is consistent
and stationary across splits.
A small fraction of points appear as sparse outliers well beyond the
central mass, corresponding to sparse but off-center flows that are large relative
to the bulk of OD entries within the same time frame but rare in
occurrence.
The cross-split consistency of both the central cluster and the
peripheral outliers indicates that this dominant-plus-sparse pattern
is a stable property of the dataset rather than an artifact of any
particular split.

Based on this observation, we adopt an additive decomposition in the
log domain:
\begin{equation}
    Z_t = U_t + O_t ,
    \label{eq:log_domain_decomposition}
\end{equation}
where $U_t$ denotes the dominant log-domain traffic structure and
$O_t$ denotes the local deviation. Since $O_t$ is expected to
satisfy the sparse pattern that most entries are close to zero while
only a few entries have large magnitudes, we impose a Laplace prior on
the deviation term:
\begin{equation}
    p(O_t) \propto \exp\left(-\lambda \|O_t\|_1\right).
    \label{eq:laplace_sparse_prior}
\end{equation}
This prior encourages sparsity through its sharp peak at zero, while
its heavy-tailed behavior allows a small number of large
deviations to be represented.

Figure~\ref{fig:lognormal_loglaplace_fit} shows marginal
distribution fits for two representative OD dimensions from the
Facebook-Pod-B dataset, within a locally stationary window of
$L = 200$ frames and $d = 56$ OD pairs.

\begin{figure}[t]
  \centering
  \includegraphics[width=\columnwidth]{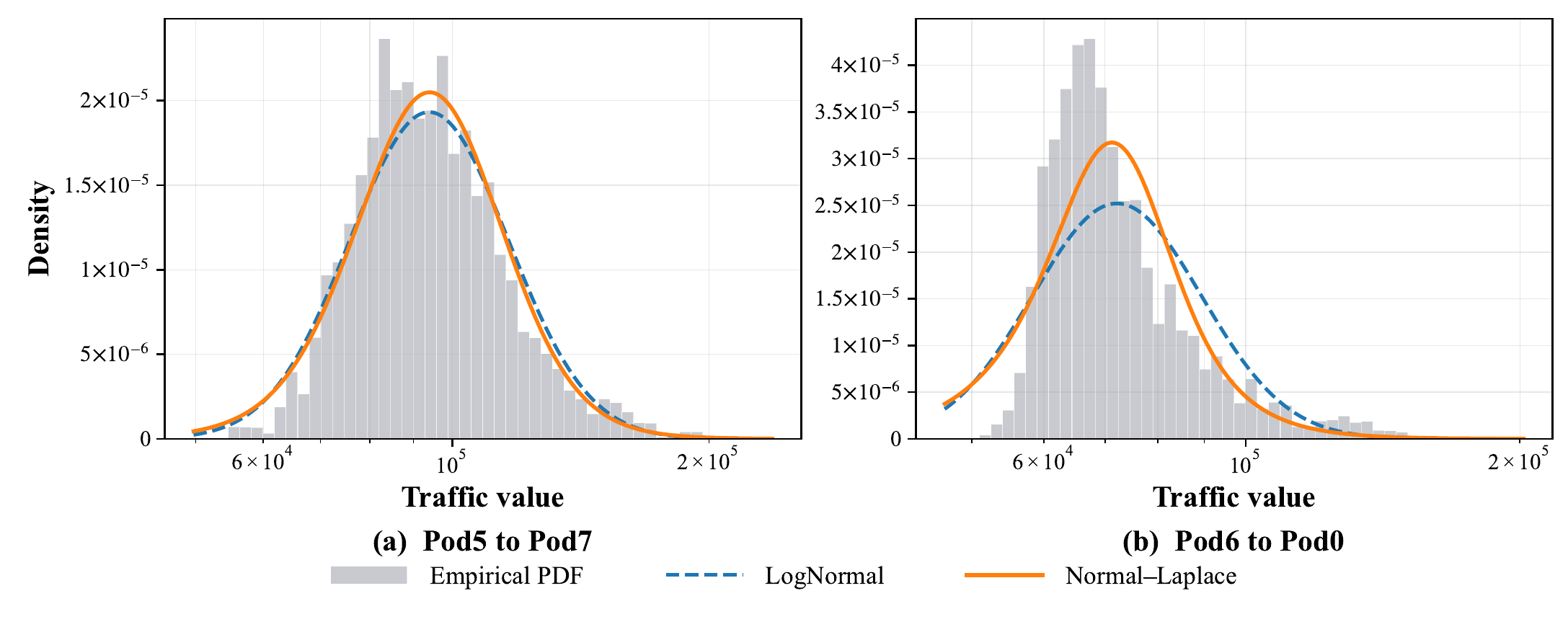}
  \caption{\textnormal{%
    Marginal distribution fitting on two representative OD
    dimensions from Facebook-Pod-B.
    (a)~Pod5$\to$Pod7: LogNormal and Normal--Laplace curves
    nearly coincide, indicating negligible sparse deviation.
    (b)~Pod6$\to$Pod0: the empirical distribution exhibits a
    sharper peak and a heavier left shift; Normal--Laplace
    provides a closer fit than LogNormal alone by capturing
    the sparse-deviation component.}}
  \label{fig:lognormal_loglaplace_fit}
\end{figure}

\section{Imputation Details}
\label{appendix:uncertainty}

This appendix derives, for each missing traffic entry, a $95\%$
imputation estimated interval that quantifies the uncertainty of
the imputed log-domain value under the estimated statistical model.
All derivations are conducted within a fixed local window $W_n$ and
a fixed time instant $t \in W_n$, conditioned on the estimated shared
parameters
$\hat{\theta}_n = (\hat{\mu}_n, \hat{\Sigma}_n, \hat{\lambda}_n)$
and the frame-level optimal deviation term $\hat{o}_t^{(\Omega_t)}$
obtained from the block coordinate descent in
Section~\ref{subsec:bcd_solver}.

For brevity, within this appendix we abbreviate
$\mathcal{O} := \Omega_t$ and
$\mathcal{M} := \Omega_t^{\mathrm{mis}}$;
these abbreviations are used only locally and do not redefine any
symbol of the main text.

\subsection{Conditional Gaussian Distribution of
             the Unobserved Principal Component}
\label{subsec:cond_gaussian}

Given the frame-level optimal deviation term
$\hat{o}_t^{(\mathcal{O})}$
returned by the deviation subproblem~\eqref{eq:ot_update_sec43},
the additive decomposition of the log-domain model
(cf.~\eqref{eq:pattern_decomposition}) constrains the principal
component on the observed positions to satisfy
\begin{equation}
  \hat{u}_t^{(\mathcal{O})}
  \;=\; z_t^{(\mathcal{O})} - \hat{o}_t^{(\mathcal{O})},
  \label{eq:u_obs}
\end{equation}
which we take as the conditioning value for the missing positions.

Under the model in~\eqref{eq:pattern_gaussian}, the principal
component satisfies
$U_t \sim \mathcal{N}(\hat{\mu}_n, \hat{\Sigma}_n)$
within $W_n$.
Applying the standard conditional distribution formula for
multivariate Gaussians~\cite{bishop2006pattern} to
\eqref{eq:u_obs} yields
\begin{equation}
U_t^{(\mathcal{M})}
  \;\Big|\; U_t^{(\mathcal{O})} = \hat{u}_t^{(\mathcal{O})}
  \;\sim\; \mathcal{N}\!\bigl(m_t,\, V_t\bigr),
  \label{eq:cond_gaussian}
\end{equation}
where the conditional mean
$m_t \in \mathbb{R}^{|\mathcal{M}|}$
and conditional covariance
$V_t \in \mathbb{R}^{|\mathcal{M}| \times |\mathcal{M}|}$
are
\begin{align}
  m_t
  &= \hat{\mu}_n^{(\mathcal{M})}
   + \hat{\Sigma}_n^{(\mathcal{M},\,\mathcal{O})}
     \bigl(\hat{\Sigma}_n^{(\mathcal{O},\,\mathcal{O})}\bigr)^{-1}
     \bigl(\hat{u}_t^{(\mathcal{O})} - \hat{\mu}_n^{(\mathcal{O})}\bigr),
  \label{eq:cond_mean}\\[4pt]
  V_t
  &= \hat{\Sigma}_n^{(\mathcal{M},\,\mathcal{M})}
   - \hat{\Sigma}_n^{(\mathcal{M},\,\mathcal{O})}
     \bigl(\hat{\Sigma}_n^{(\mathcal{O},\,\mathcal{O})}\bigr)^{-1}
     \hat{\Sigma}_n^{(\mathcal{O},\,\mathcal{M})}.
  \label{eq:cond_cov}
\end{align}

The subvectors and submatrices in~\eqref{eq:cond_mean}--\eqref{eq:cond_cov}
follow the notation of~\eqref{eq:Ut_obs_gaussian}, with
$\mathcal{O}$ and $\mathcal{M}$ substituted for $\Omega_t$ and
$\Omega_t^{\mathrm{mis}}$.

\subsection{Laplace Distribution for the Unobserved
             Deviation Component}
\label{subsec:laplace_prior}

With the distribution of the unobserved principal component
established in~\eqref{eq:cond_gaussian}, we next characterise the
distribution of the unobserved deviation component.

Under the model in~\eqref{eq:pattern_laplacian}, each entry of
$O_t$ follows a univariate density proportional to
$\exp(-\lambda_n |o|)$.
Comparing this with the standard parameterisation of
$\mathrm{Laplace}(0,b)$, whose density is
$\tfrac{1}{2b}\exp(-|x|/b)$, identifies the scale parameter as the
reciprocal of $\lambda_n$.
Substituting the estimate $\hat{\lambda}_n$ therefore gives
\begin{equation}
  b \;:=\; \frac{1}{\hat{\lambda}_n} \;>\; 0,
  \label{eq:def_b}
\end{equation}
so that the deviation component on the missing positions has the
marginal distribution
\begin{equation}
  O_t^{(\mathcal{M})}
  \;\sim\;
  \prod_{j \in \mathcal{M}}
    \mathrm{Laplace}(0,\, b).
  \label{eq:laplace_missing}
\end{equation}

Furthermore, since $U_t \perp O_t$ for each $t$ under the model in
Section~\ref{subsec:system_model}, this independence is preserved
under conditioning:
the conditional variable in~\eqref{eq:cond_gaussian} remains
independent of~\eqref{eq:laplace_missing}.

\subsection{Conditional estimated Density for
             Unobserved Positions}
\label{subsec:convolution_density}

Combining the distributions in~\eqref{eq:cond_gaussian}
and~\eqref{eq:laplace_missing} through the additive decomposition
in~\eqref{eq:pattern_decomposition}, the log-domain variable on the
missing positions satisfies
\begin{equation}
  Z_t^{(\mathcal{M})}
  \;=\; U_t^{(\mathcal{M})} + O_t^{(\mathcal{M})}.
  \label{eq:decomp_missing}
\end{equation}

By the independence established in
Section~\ref{subsec:laplace_prior}, the joint conditional estimated
density of $Z_t^{(\mathcal{M})}$ is the convolution of the Gaussian
in~\eqref{eq:cond_gaussian} and the product Laplace
in~\eqref{eq:laplace_missing}:
\begin{equation}
\begin{split}
  & p\!\left(z^{(\mathcal{M})}
    \;\Big|\; z_t^{(\mathcal{O})},\,
    \hat{o}_t^{(\mathcal{O})},\,\hat{\theta}_n\right) = \\
  & \qquad \int_{\mathbb{R}^{|\mathcal{M}|}}
    \mathcal{N}(u \mid m_t,\, V_t)
    \prod_{j \in \mathcal{M}}
      \frac{1}{2b}\,
      e^{-|z_j - u_j|/b}
    \,\mathrm{d}u,
\end{split}
\label{eq:joint_convolution}
\end{equation}
where $\mathcal{N}(u \mid m_t, V_t)$ denotes the Gaussian density
with mean $m_t$ and covariance $V_t$ evaluated at $u$.

\subsection{Per-Coordinate Marginal Distribution and
             Closed-Form CDF}
\label{subsec:marginal_cdf}

Since the $95\%$ interval is constructed separately for each missing
coordinate, we reduce~\eqref{eq:joint_convolution} to its
one-dimensional marginals.
For each $j \in \mathcal{M}$, the marginal at coordinate $j$
decomposes as
\begin{equation}
  Z_t(j) \;=\; U_t(j) + O_t(j),
  \label{eq:marginal_decomp}
\end{equation}
where the two components are independent and distributed as
\begin{equation}
\begin{split}
  U_t(j) \;\Big|\; U_t^{(\mathcal{O})} = \hat{u}_t^{(\mathcal{O})}
  &\;\sim\; \mathcal{N}(m_{t,j},\, s_{t,j}^2), \\
  O_t(j) &\;\sim\; \mathrm{Laplace}(0,\, b),
\end{split}
\label{eq:marginal_dists}
\end{equation}
with the marginal parameters defined as
\begin{equation}
  m_{t,j} \;:=\; (m_t)_j,
  \qquad
  s_{t,j}^2 \;:=\; (V_t)_{jj}.
  \label{eq:marginal_params}
\end{equation}

Consequently, $Z_t(j)$ follows the convolution of a Gaussian and a
Laplace distribution~\cite{reed2006normal}, with variance
\begin{equation}
  \mathrm{Var}\!\bigl(Z_t(j)\bigr)
  \;=\; s_{t,j}^2 + 2b^2,
  \label{eq:nl_variance}
\end{equation}
where we used the independence of $U_t(j)$ and $O_t(j)$ together with
the identity $\mathrm{Var}(\mathrm{Laplace}(0,b)) = 2b^2$.

\medskip

To present the closed-form CDF compactly, introduce the auxiliary
quantities
\begin{equation}
  \alpha_{t,j}(z) \;:=\; \frac{z - m_{t,j}}{s_{t,j}},
  \qquad
  \kappa_{t,j} \;:=\; \frac{s_{t,j}}{b}.
  \label{eq:alpha_kappa}
\end{equation}

\begin{proposition}[Closed-form CDF]
\label{prop:nl_cdf}
The CDF of $Z_t(j)$,
\begin{equation}
  F_{t,j}(z)
  \;:=\;
  P\!\left(Z_t(j) \leq z
    \;\Big|\; z_t^{(\mathcal{O})},\,\hat{o}_t^{(\mathcal{O})},\,
    \hat{\theta}_n\right),
  \label{eq:cdf_def}
\end{equation}
admits the closed form
\begin{align}
  F_{t,j}(z)
  &= \Phi\!\bigl(\alpha_{t,j}(z)\bigr)
  \notag\\[4pt]
  &\quad
  -\;\frac{1}{2}\exp\!\!\left(
    \frac{s_{t,j}^2}{2b^2} - \frac{z - m_{t,j}}{b}
  \right)
  \Phi\!\bigl(\alpha_{t,j}(z) - \kappa_{t,j}\bigr)
  \notag\\[4pt]
  &\quad
  +\;\frac{1}{2}\exp\!\!\left(
    \frac{s_{t,j}^2}{2b^2} + \frac{z - m_{t,j}}{b}
  \right)
  \Phi\!\bigl(-\alpha_{t,j}(z) - \kappa_{t,j}\bigr),
  \label{eq:nl_cdf}
\end{align}
where $\Phi(\cdot)$ is the standard normal CDF.
\end{proposition}

\begin{proof}
By the total probability formula and the Laplace density of
$O_t(j)$,
\begin{equation}
  F_{t,j}(z)
  = \int_{-\infty}^{+\infty}
      \Phi\!\!\left(\frac{z - o - m_{t,j}}{s_{t,j}}\right)
      \frac{1}{2b}\,e^{-|o|/b}\,\mathrm{d}o.
  \label{eq:cdf_integral}
\end{equation}

Split the integration domain at $o = 0$:
$F_{t,j}(z) = I_1(z) + I_2(z)$,
where $I_1$ and $I_2$ denote the integrals over $[0,+\infty)$ and
$(-\infty, 0)$, respectively.

\smallskip
\noindent\textit{Evaluation of $I_1$.}\;
Applying integration by parts with
$u_1 = \Phi\!\bigl(\alpha_{t,j}(z) - o/s_{t,j}\bigr)$
and $\mathrm{d}v_1 = e^{-o/b}\,\mathrm{d}o$,
\begin{equation}
\begin{split}
  2b\,I_1(z)
  &= \int_0^{+\infty}
       \Phi\!\bigl(\alpha_{t,j}(z) - o/s_{t,j}\bigr)\,
       e^{-o/b}\,\mathrm{d}o \\
  &= b\,\Phi\!\bigl(\alpha_{t,j}(z)\bigr)
     - \frac{b}{s_{t,j}}\, J_1(z),
\end{split}
\label{eq:I1_ibp}
\end{equation}
where
\begin{equation}
  J_1(z) \;:=\; \int_0^{+\infty}
    \phi\!\bigl(\alpha_{t,j}(z) - o/s_{t,j}\bigr)\,
    e^{-o/b}\,\mathrm{d}o
  \label{eq:J1_def}
\end{equation}
and $\phi(\cdot)$ denotes the standard normal density.

Completing the square in the combined exponent of $J_1(z)$ gives
\begin{equation}
\begin{split}
  & -\tfrac{1}{2}\bigl(\alpha_{t,j}(z) - o/s_{t,j}\bigr)^2 - \tfrac{o}{b} \\
  &\qquad =\;
  -\frac{\bigl[o - s_{t,j}\bigl(\alpha_{t,j}(z) - \kappa_{t,j}\bigr)\bigr]^2}
        {2 s_{t,j}^2}
  + \frac{s_{t,j}^2}{2b^2}
  - \frac{z - m_{t,j}}{b}.
\end{split}
\label{eq:sq_I1}
\end{equation}

Substituting
$w = \bigl[o - s_{t,j}(\alpha_{t,j}(z) - \kappa_{t,j})\bigr]/s_{t,j}$
so that $o = 0$ corresponds to
$w = \kappa_{t,j} - \alpha_{t,j}(z)$,
and noting that the new lower limit yields the standard normal tail
probability,
\begin{equation}
  J_1(z)
  \;=\;
  \exp\!\!\left(\frac{s_{t,j}^2}{2b^2} - \frac{z - m_{t,j}}{b}\right)
  \Phi\!\bigl(\alpha_{t,j}(z) - \kappa_{t,j}\bigr).
  \label{eq:J1_eval}
\end{equation}

Substituting \eqref{eq:J1_eval} into \eqref{eq:I1_ibp} and
simplifying,
\begin{equation}
\begin{split}
  I_1(z)
  &= \tfrac{1}{2}\,\Phi\!\bigl(\alpha_{t,j}(z)\bigr) \\
  &\quad
  -\;\tfrac{1}{2}
    \exp\!\!\left(\frac{s_{t,j}^2}{2b^2} - \frac{z - m_{t,j}}{b}\right)
    \Phi\!\bigl(\alpha_{t,j}(z) - \kappa_{t,j}\bigr).
\end{split}
\label{eq:I1}
\end{equation}

\smallskip
\noindent\textit{Evaluation of $I_2$.}\;
Substituting $o \mapsto -o$ maps $I_2$ to an integral over
$[0,+\infty)$ whose integrand differs from that of $I_1$ only in
the sign of $o/s_{t,j}$.
Applying the same integration-by-parts, completing-the-square, and
substitution procedure as for $I_1$, with this sign reversed
throughout, yields
\begin{equation}
\begin{split}
  I_2(z)
  &= \tfrac{1}{2}\,\Phi\!\bigl(\alpha_{t,j}(z)\bigr) \\
  &\quad
  +\;\tfrac{1}{2}
    \exp\!\!\left(\frac{s_{t,j}^2}{2b^2} + \frac{z - m_{t,j}}{b}\right)
    \Phi\!\bigl(-\alpha_{t,j}(z) - \kappa_{t,j}\bigr).
\end{split}
\label{eq:I2}
\end{equation}

Summing \eqref{eq:I1} and \eqref{eq:I2} yields \eqref{eq:nl_cdf}.

\end{proof}

\subsection{Imputation Estimated Interval}
\label{subsec:estimated_interval}

Equipped with the closed-form marginal CDF $F_{t,j}$ from
Proposition~\ref{prop:nl_cdf}, we now construct a $97.5\%$ interval
for each missing coordinate in the log domain and map it to the
original traffic domain.

The marginal distributions in~\eqref{eq:marginal_dists} are each
symmetric about their respective centres.
By the independence of $U_t(j)$ and $O_t(j)$, the sum
$Z_t(j) - m_{t,j}$ is therefore also symmetric about zero.

Consequently, the $97.5\%$ interval in the log domain, defined as the
interval between the $1.25\%$ and $98.75\%$ quantiles of $Z_t(j)$,
takes the symmetric form
\begin{equation}
  \mathrm{PI}^{(z)}_{0.975}(t,j)
  \;=\;
  \bigl[m_{t,j} - \delta_{t,j},\; m_{t,j} + \delta_{t,j}\bigr],
  \label{eq:log_pi}
\end{equation}
where $\delta_{t,j} > 0$ is the unique solution to
\begin{equation}
  F_{t,j}\!\left(m_{t,j} + \delta_{t,j}\right) \;=\; 0.9875.
  \label{eq:delta_eq}
\end{equation}

By the strictly increasing transformation
in~\eqref{eq:inverse_transform}, quantile ordering is preserved, and
the $97.5\%$ imputation estimated interval in the original traffic
domain is
\begin{equation}
  \mathrm{PI}^{(x)}_{0.975}(t,j)
  \;=\;
  \Bigl[
    e^{m_{t,j} - \delta_{t,j}} - \varepsilon_j,\;\;
    e^{m_{t,j} + \delta_{t,j}} - \varepsilon_j
  \Bigr],
  \label{eq:original_pi}
\end{equation}
where $\varepsilon_j$ is the $j$-th component of the offset vector
$\varepsilon$ in~\eqref{eq:Zt_def}.

Equation~\eqref{eq:original_pi} serves as the per-entry uncertainty
output of Utimac and is evaluated in
Section~\ref{sec:experiments}.

\begin{figure}[htbp]
  \centering
  \includegraphics[width=\columnwidth]{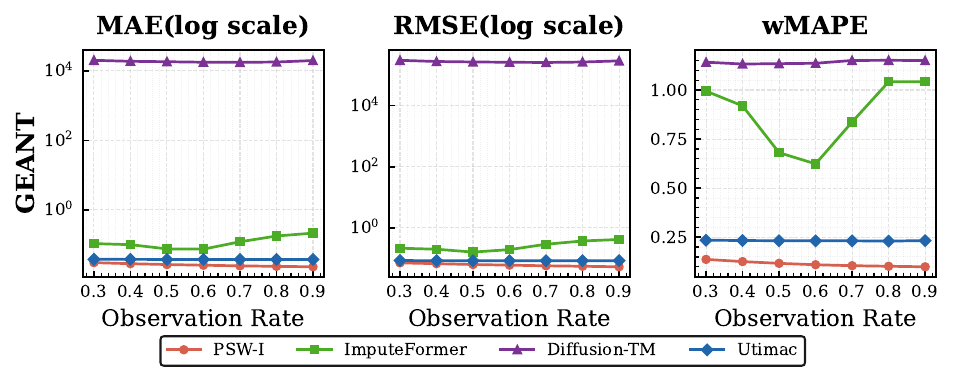}
  \caption{\textnormal{Overall MAE, RMSE, and wMAPE vs.\
    $p_{\mathrm{obs}}$ on GÉANT.}}
  \label{fig:geant_overall}
\end{figure}

\begin{figure}[htbp]
  \centering
  \includegraphics[width=\columnwidth]{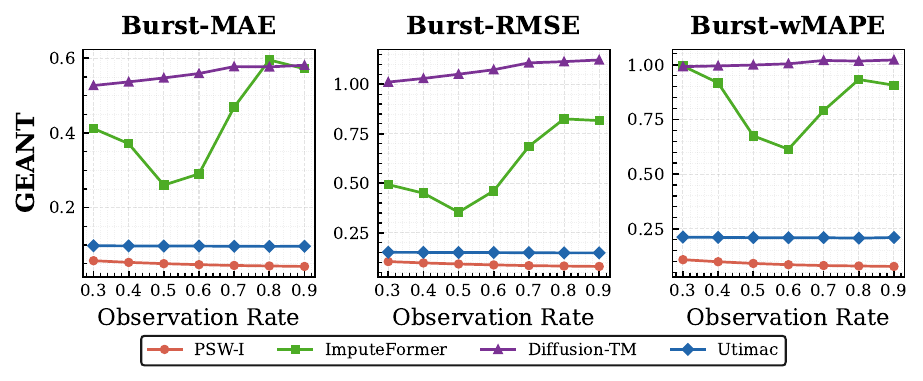}
  \caption{\textnormal{Burst-MAE, Burst-RMSE, and Burst-wMAPE
    vs.\ $p_{\mathrm{obs}}$ on GÉANT.}}
  \label{fig:geant_burst}
\end{figure}

%
 
\section{Dataset Statistics}
\label{app:datasets}
 
Table~\ref{tab:datasets} summarises the three datasets used in our
experiments.
GÉANT~\cite{uhlig2006providing} is a wide-area network (WAN) dataset
from the pan-European GÉANT backbone.
Facebook-Pod-B, and Facebook-ToR-A are
datacenter network (DCN) datasets from Facebook's production
infrastructure, representing pod-level and
rack-level traffic respectively.
 
\begin{table}[h]
  \centering
  \caption{Dataset statistics ($^\dagger$WAN; $^\ddagger$DCN).}
  \label{tab:datasets}
  \setlength{\tabcolsep}{4pt}
  \renewcommand{\arraystretch}{1.15}
  \begin{tabular}{lrrrr}
    \toprule
    \textbf{Dataset} & $N$ & \textbf{Links} & \textbf{OD flows} & $d$ \\
    \midrule
    GÉANT$^\dagger$        & 22  & 72      & 462      & 484      \\
    Facebook-Pod-B$^\ddagger$  &  8  & 56      &  56      &  64      \\
    Facebook-ToR-A$^\ddagger$  & 155 & 7{,}194 & 23{,}870 & 24{,}025 \\
    \bottomrule
  \end{tabular}
\end{table}
 
\section{Evaluation Metric Definitions}
\label{app:metrics}
 
 
 
All metrics are computed exclusively on missing positions
$\Omega^{\mathrm{miss}}_t = \{k \mid B_t(k){=}0\}$, where $B_t$
is the observation mask at time $t$.
Let $\hat{x}_t(k)$ and $x_t(k)$ denote the imputed and true
traffic values at position $k$ and time $t$, and let
$\varepsilon{>}0$ be a small constant preventing division by zero.
 
\subsection*{Overall Imputation Accuracy}
 
\begin{align}
  \mathrm{MAE}
  &= \frac{\displaystyle\sum_t\sum_{k\in\Omega^{\mathrm{miss}}_t}
           |\hat{x}_t(k)-x_t(k)|}
          {\displaystyle\sum_t|\Omega^{\mathrm{miss}}_t|} ,
  \label{eq:mae} \\[4pt]
  \mathrm{RMSE}
  &= \sqrt{\frac{\displaystyle\sum_t\sum_{k\in\Omega^{\mathrm{miss}}_t}
                 (\hat{x}_t(k)-x_t(k))^2}
                {\displaystyle\sum_t|\Omega^{\mathrm{miss}}_t|}} ,
  \label{eq:rmse} \\[4pt]
  \mathrm{wMAPE}
  &= \frac{\displaystyle\sum_t\sum_{k\in\Omega^{\mathrm{miss}}_t}
           |\hat{x}_t(k)-x_t(k)|}
          {\displaystyle\sum_t\sum_{k\in\Omega^{\mathrm{miss}}_t}
           |x_t(k)|{+}\varepsilon} .
  \label{eq:wmape}
\end{align}
wMAPE is preferred over per-entry MAPE because it aggregates
numerator and denominator separately, making it robust to near-zero
flows common in heavy-tailed traffic distributions.

\subsection*{Burst Flow Definition}
 
A flow $k$ at time $t$ is defined as a \emph{burst flow} if its
dominance ratio $r_t(k;\alpha) \ge \beta$, where
\begin{equation}
  r_t(k;\alpha) = \frac{1}{d{-}1}
  \sum_{\ell \ne k}
  \mathbf{1}\!\bigl[x_t(k) \ge \alpha\,x_t(\ell)\bigr] ,
  \label{eq:burst_def}
\end{equation}
with dominance multiplier $\alpha{>}1$ and majority threshold
$\beta{\in}(0,1)$.
Let $\mathcal{B}^{(\alpha,\beta)}_t$ denote the burst flow set at
time $t$, and define the missing burst set as
$\mathcal{B}^{\mathrm{miss}}_t =
\mathcal{B}^{(\alpha,\beta)}_t \cap \Omega^{\mathrm{miss}}_t$,
with predicted counterpart $\hat{\mathcal{B}}^{\mathrm{miss}}_t$
obtained by applying the same definition to $\hat{x}_t$.

\subsection*{Burst Flow Detection}
 
\begin{align}
  \mathrm{Precision}
  &= \frac{\displaystyle\sum_t
           |\hat{\mathcal{B}}^{\mathrm{miss}}_t
             \cap \mathcal{B}^{\mathrm{miss}}_t|}
          {\displaystyle\sum_t
           |\hat{\mathcal{B}}^{\mathrm{miss}}_t|
           {+}\,\varepsilon} ,
  \label{eq:precision} \\[4pt]
  \mathrm{Recall}
  &= \frac{\displaystyle\sum_t
           |\hat{\mathcal{B}}^{\mathrm{miss}}_t
             \cap \mathcal{B}^{\mathrm{miss}}_t|}
          {\displaystyle\sum_t
           |\mathcal{B}^{\mathrm{miss}}_t|
           {+}\,\varepsilon} ,
  \label{eq:recall} \\[4pt]
  \mathrm{F1}
  &= \frac{2\cdot\mathrm{Precision}\cdot\mathrm{Recall}}
          {\mathrm{Precision}+\mathrm{Recall}+\varepsilon} .
  \label{eq:f1}
\end{align}

Recall and F1 are the primary indicators: Recall measures coverage
of true burst flows, while F1 balances false alarms and missed
detections.

\subsection*{Burst Flow Imputation Accuracy}
 
Evaluated only on true missing burst positions
$\mathcal{B}^{\mathrm{miss}}_t$:
\begin{align}
  \mathrm{Burst\text{-}MAE}
  &= \frac{\displaystyle\sum_t\sum_{k\in\mathcal{B}^{\mathrm{miss}}_t}
           |\hat{x}_t(k)-x_t(k)|}
          {\displaystyle\sum_t|\mathcal{B}^{\mathrm{miss}}_t|} ,
  \label{eq:burst_mae} \\[4pt]
  \mathrm{Burst\text{-}RMSE}
  &= \sqrt{\frac{\displaystyle\sum_t
                 \sum_{k\in\mathcal{B}^{\mathrm{miss}}_t}
                 (\hat{x}_t(k)-x_t(k))^2}
                {\displaystyle\sum_t
                 |\mathcal{B}^{\mathrm{miss}}_t|}} ,
  \label{eq:burst_rmse} \\[4pt]
  \mathrm{Burst\text{-}wMAPE}
  &= \frac{\displaystyle\sum_t
           \sum_{k\in\mathcal{B}^{\mathrm{miss}}_t}
           |\hat{x}_t(k)-x_t(k)|}
          {\displaystyle\sum_t
           \sum_{k\in\mathcal{B}^{\mathrm{miss}}_t}
           |x_t(k)|{+}\varepsilon} .
  \label{eq:burst_wmape}
\end{align}

Burst-wMAPE is the headline indicator as it directly reflects
relative recovery quality on high-volume critical entries.


\section{Extended Experimental Results}
\label{app:extended}

This appendix presents two sets of supplementary results omitted
from the main text for space: (i) complete imputation performance
on the GÉANT WAN dataset with the same three figures as
Section~\ref{subsec:results}, and (ii) burst flow detection
metrics (Precision, Recall, F1) for all three datasets.

\subsection{GÉANT Results}
\label{app:geant}

Figures~\ref{fig:geant_overall}--\ref{fig:geant_burst} report
the same three figures as in the main text, now for the GÉANT
WAN backbone dataset ($N{=}22$, $d{=}484$, 15-min intervals).


\begin{figure}[t]
  \centering
  \includegraphics[width=\columnwidth]{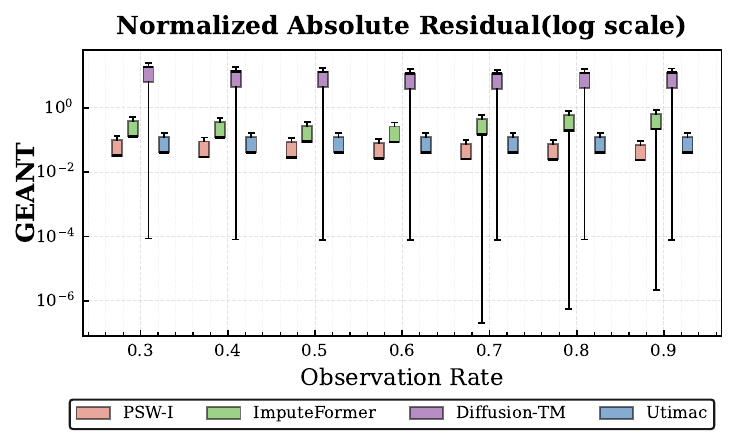}
  \caption{\textnormal{Normalized absolute residual distributions
    (5th--95th percentile) vs.\ $p_{\mathrm{obs}}$ on GÉANT.}}
  \label{fig:geant_residual}
\end{figure}


On GÉANT, the ranking of methods differs from the DCN setting in
an instructive way.
For overall imputation, PSW-I achieves the lowest MAE across the
full observation range (e.g., $0.030$ at $p_{\mathrm{obs}}{=}0.3$),
while Utimac ranks second with MAE $= 0.037$, staying well below
ImputeFormer (up to $0.211$ at $p_{\mathrm{obs}}{=}0.9$) and far
below Diffusion-TM, whose MAE collapses to values exceeding
$10^4$ due to numerical instability in the per-dimension
min-max normalization applied to WAN-scale data.
In terms of wMAPE, Utimac achieves $0.235$ at
$p_{\mathrm{obs}}{=}0.3$ and remains stable around $0.231$--$0.235$
across the full range, while ImputeFormer degrades from $0.994$ to
$1.042$ and Diffusion-TM stays above $1.13$ throughout.
These results demonstrate that Utimac generalizes well beyond the
DCN setting: even on a WAN backbone with qualitatively different
traffic characteristics, it produces stable and competitive
completions, trailing only PSW-I (an alignment-based method
that does not require traffic-specific structural assumptions)
while substantially outperforming the two neural baselines.

The residual distribution in Figure~\ref{fig:geant_residual}
reinforces this picture: Utimac's 95th-percentile residual
($0.164$ at $p_{\mathrm{obs}}{=}0.3$) is close to PSW-I
($0.131$) and far below Diffusion-TM, whose extreme tail
residuals (p95 up to $24.97$) confirm the numerical instability
visible in Figure~\ref{fig:geant_overall}.

For burst flow imputation on GÉANT
(Figure~\ref{fig:geant_burst}), Utimac achieves a stable
Burst-wMAPE around $0.209$--$0.212$ across all observation rates,
ranking second behind PSW-I ($0.077$--$0.108$) and substantially
ahead of ImputeFormer ($0.674$--$0.994$) and Diffusion-TM
($0.991$--$1.022$).
The gap to PSW-I on GÉANT reflects the particular advantage of
alignment-based methods on low-dimensional WAN traffic matrices
($d{=}484$) where distributional proximity can be computed
accurately with small batches, an advantage that diminishes as
dimensionality grows.

\subsection{Burst Flow Detection}
\label{app:burst_detection}

Tables~\ref{tab:burst_tora} and~\ref{tab:burst_geant} report
Precision, Recall, and F1 for burst flow detection
($\alpha{=}2.0$, $\beta{=}0.8$) on Facebook-ToR-A and GÉANT.
Burst flow detection differs from burst flow imputation
(Section~\ref{subsec:results}): detection measures whether
a method correctly identifies which flows are bursty based
on its imputed values, independent of the numerical recovery
error.

\begin{table}[htbp]
  \centering
  \caption{Burst flow detection on Facebook-ToR-A.
    Precision ${\approx}1.000$ for all methods (see text).$^{\dagger}$}
  \label{tab:burst_tora}
  \setlength{\tabcolsep}{4pt}
  \renewcommand{\arraystretch}{1.15}
  \begin{tabular}{lccccccc}
    \toprule
    & \multicolumn{7}{c}{Observation rate $p_{\mathrm{obs}}$} \\
    \cmidrule(lr){2-8}
    Method & 0.3 & 0.4 & 0.5 & 0.6 & 0.7 & 0.8 & 0.9 \\
    \midrule
    \multicolumn{8}{l}{\textit{Recall}} \\
    PSW-I         & .202 & .199 & .201 & .217 & .257 & .345 & .474 \\
    ImputeFormer  & .149 & .160 & .168 & .174 & .178 & .182 & .185 \\
    Diffusion-TM  & .196 & .205 & .217 & .234 & .260 & .304 & .437 \\
    Utimac        & .151 & .157 & .167 & .186 & .220 & .282 & \textbf{1.000} \\
    \midrule
    \multicolumn{8}{l}{\textit{F1}} \\
    PSW-I         & .337 & .332 & .334 & .357 & .409 & .513 & .643 \\
    ImputeFormer  & .259 & .276 & .287 & .296 & .302 & .307 & .312 \\
    Diffusion-TM  & .328 & .340 & .356 & .379 & .412 & .466 & .608 \\
    Utimac        & .262 & .271 & .287 & .314 & .361 & .440 & \textbf{1.000} \\
    \bottomrule
  \end{tabular}
  \smallskip\\
  {\footnotesize $^{\dagger}$Precision\,${\approx}1.000$ for all methods due to extreme burst-flow
  sparsity at $d{=}24{,}025$; Recall and F1 are the sole discriminative indicators.}
\end{table}

\begin{table}[htbp]
  \centering
  \caption{PICP$_{97.5}$ (\%) of Utimac vs.\ observation rate
    $p_{\mathrm{obs}}$ on Facebook-Pod-B and Facebook-ToR-A.
    Values close to 97.5 indicate estimated
    intervals.}
  \label{tab:picp975}
  \setlength{\tabcolsep}{4pt}
  \renewcommand{\arraystretch}{1.15}
  \begin{tabular}{lccccccc}
    \toprule
    & \multicolumn{7}{c}{Observation rate $p_{\mathrm{obs}}$} \\
    \cmidrule(lr){2-8}
    Dataset & 0.3 & 0.4 & 0.5 & 0.6 & 0.7 & 0.8 & 0.9 \\
    \midrule
    Facebook-Pod-B  & .968 & .951 & .955 & .958 & .959 & .966 & .973 \\
    Facebook-ToR-A  & .970 & .972 & .973 & .974 & .975 & .975 & .976 \\
    \bottomrule
  \end{tabular}
\end{table}
\label{app:picp}

\begin{table}[htbp]
  \centering
  \caption{Burst flow detection on GÉANT.}
  \label{tab:burst_geant}
  \setlength{\tabcolsep}{3.5pt}
  \renewcommand{\arraystretch}{1.15}
  \begin{tabular}{lcccccccc}
    \toprule
    & \multicolumn{7}{c}{Observation rate $p_{\mathrm{obs}}$} \\
    \cmidrule(lr){2-8}
    Method & 0.3 & 0.4 & 0.5 & 0.6 & 0.7 & 0.8 & 0.9 \\
    \midrule
    \multicolumn{8}{l}{\textit{Precision}} \\
    PSW-I         & .924 & .930 & .934 & .938 & .942 & .944 & .948 \\
    ImputeFormer  & .606 & .639 & .694 & .787 & .883 & .919 & .923 \\
    Diffusion-TM  & .583 & .592 & .590 & .587 & .583 & .580 & .586 \\
    Utimac        & \textbf{.920} & \textbf{.923} & \textbf{.928} & \textbf{.930} & \textbf{.932} & \textbf{.935} & \textbf{.934} \\
    \midrule
    \multicolumn{8}{l}{\textit{Recall}} \\
    PSW-I         & .921 & .924 & .928 & .933 & .935 & .936 & .941 \\
    ImputeFormer  & .851 & .905 & .927 & .925 & .901 & .867 & .802 \\
    Diffusion-TM  & .544 & .543 & .534 & .524 & .519 & .515 & .523 \\
    Utimac        & \textbf{.894} & \textbf{.891} & \textbf{.889} & \textbf{.886} & \textbf{.883} & \textbf{.880} & \textbf{.875} \\
    \midrule
    \multicolumn{8}{l}{\textit{F1}} \\
    PSW-I         & .922 & .927 & .931 & .935 & .939 & .940 & .945 \\
    ImputeFormer  & .708 & .749 & .794 & .851 & .892 & .892 & .858 \\
    Diffusion-TM  & .563 & .566 & .561 & .554 & .549 & .546 & .553 \\
    Utimac        & \textbf{.907} & \textbf{.907} & \textbf{.908} & \textbf{.907} & \textbf{.907} & \textbf{.907} & \textbf{.903} \\
    \bottomrule
  \end{tabular}
\end{table}

On \textbf{Facebook-ToR-A} (Table~\ref{tab:burst_tora}),
all methods obtain Precision ${\approx} 1.000$, a statistical
artifact of the dataset's extreme sparsity ($d{=}24{,}025$): at
this dimensionality, burst flows constitute a negligibly small
fraction of all non-zero entries, so virtually all predicted burst
positions are genuine.
Recall is therefore the sole discriminative indicator.
At $p_{\mathrm{obs}}{=}0.9$, Utimac's Recall jumps to $\mathbf{1.000}$
while PSW-I reaches $0.474$ and Diffusion-TM $0.437$, revealing
that Utimac's inference framework can recover the full burst-flow
structure when sufficient observations are available even in a
high-dimensional sparse setting.

On \textbf{GÉANT} (Table~\ref{tab:burst_geant}), burst detection
is considerably easier than on the DCN datasets due to the
lower dimensionality.
PSW-I leads with F1 ${\approx} 0.92$--$0.94$, followed closely by
Utimac at F1 ${\approx} 0.90$--$0.91$, a gap of only $1.7$--$4.4\%$.
Diffusion-TM achieves the weakest detection performance with
F1 $\approx 0.55$, consistent with its numerical instability
on GÉANT observed in Section~\ref{app:geant}.

\subsection{Estimated Interval Coverage}

Beyond point estimates, Utimac derives closed-form estimated
intervals for every missing entry from the analytically
characterised Gaussian--Laplace posterior
(Section~\ref{sec:method}).
We evaluate calibration via the \textbf{Prediction Interval
Coverage Probability at the 97.5\% nominal level}
(PICP$_{97.5}$), defined as the fraction of true missing values
that fall within the predicted 97.5\% interval.
A well-calibrated method should achieve
PICP$_{97.5} \approx 97.5\%$ regardless of the observation rate.

Table~\ref{tab:picp975} reports PICP$_{97.5}$ on Facebook-Pod-B
and Facebook-ToR-A across all observation rates.
On Facebook-ToR-A, coverage remains tightly around $97.5\%$
across the full range of $p_{\mathrm{obs}}$, confirming that the
posterior-derived intervals faithfully reflect the true uncertainty
of the imputed values.
On Facebook-Pod-B, coverage increases monotonically from $96.8\%$
at $p_{\mathrm{obs}}{=}0.3$ to $97.3\%$ at $p_{\mathrm{obs}}{=}0.9$:
at very low observation rates the posterior uncertainty is larger
and the interval width grows accordingly, so the nominal level is
approached from below as more observations constrain the inference.
This behaviour is consistent with the theoretical guarantee that
the interval endpoints are computed from the exact marginal
covariance; the convergence to the nominal level as
$p_{\mathrm{obs}}$ increases provides direct empirical support for
the correctness of the closed-form estimated distribution.

\end{document}